\pdfoutput=1
\documentclass[12pt]{article}

 \usepackage[top=2.5cm,left=2.5cm,right=2.5cm,foot=1.5cm,bottom=2.5cm]{geometry}
 \usepackage{graphicx}
 \usepackage{amssymb}
 \usepackage{amsmath}
 \usepackage{amsthm}
 \usepackage{mathtools}
 \usepackage{bbold}
 \usepackage{amsmath, amsthm, amssymb}
 \usepackage{bbm}
 \usepackage{xcolor}
 \usepackage[authoryear,round]{natbib}
\usepackage{caption}
 \usepackage{subcaption,tikz}
 \usetikzlibrary{decorations.pathreplacing}

 \newtheorem{lemma}{Lemma}
 \newtheorem{theorem}{Theorem}
 
 \newtheorem{proposition}{Proposition}

 \newtheorem{corollary}{Corollary}

 \newcommand\norm[1]{\left\lVert#1\right\rVert}

 \usepackage{esint}
 \usepackage{hyperref}
 \definecolor{darkblue}{rgb}{0.0,0.0,0.6}
 \hypersetup{
 	colorlinks=true, %set true if you want colored links
 	linktoc=all,     %set to all if you want both sections and subsections linked
 	linkcolor=darkblue,  %choose some color if you want links to stand out
 	citecolor=darkblue,
 	urlcolor=darkblue,
 }
 \allowdisplaybreaks[1]
 \title{Equilibrium Stability and Uniqueness with a Large Number of Commodities and Patient Consumers}
\author{Xinyang Wang\thanks{Email: \href{mailto:xinyang.wang@itam.mx}{xinyang.wang@itam.mx}.}\\ \small The Center for Economic Research, ITAM,\\ \small Av. Camino Santa Teresa 930, Mexico City, 10700, Mexico}
 \date{}
 
 \usepackage[T1]{fontenc}
 \usepackage{tgpagella}

\begin{document}
 
\maketitle

\begin{abstract}

We show that a large effective number of commodities can be a source of equilibrium stability and uniqueness: expanding substitution opportunities strengthens aggregate substitution effects. We study finite dated-commodity exchange economies obtained by truncating a countably infinite-horizon environment with discounted, additively separable utilities. In this setting, the effective number of commodities is the discounted count of dated commodities, so greater patience makes distant commodities more relevant. With an appropriate normalization, equilibrium substitution effects accumulate at the rate of the effective number of commodities. When a preference diversification condition holds, equilibrium income effects grow at a lower rate. The condition is satisfied, for example, when agents have sparse or localized taste differences across commodities, or when their taste profiles become sufficiently heterogeneous as the commodity space expands. Hence, whenever the effective number of commodities is sufficiently large, every equilibrium is locally tâtonnement stable, which in turn implies equilibrium uniqueness.

\end{abstract}

\noindent\textbf{JEL Classification:} C62; D50; D51\\

\noindent\textbf{Keywords:} tâtonnement stability; equilibrium uniqueness; number of commodities; substitution effects; preference diversification

\section{Introduction}

T\^atonnement stability has been a central concern of general equilibrium theory since its early development. In a continuous-time t\^atonnement process, prices move in the direction of excess demand, and an equilibrium price is said to be locally t\^atonnement stable if every trajectory starting sufficiently close to it converges back to that equilibrium. Although \citet*{arrow1959toward} emphasized that t\^atonnement is not a fully satisfactory disequilibrium model of price adjustment, equilibrium stability remains an attractive property. It accords with the basic intuition that small perturbations should not drive the economy away from equilibrium and, through index theory, local stability at all equilibria implies equilibrium uniqueness, thereby giving a logical foundation to the common guess and verify approach to equilibrium analysis.

While it may be tempting to view equilibrium stability as a natural property of general equilibrium models, \citet*{scarf1960some} showed with a simple example that this need not hold. More fundamentally, the Sonnenschein--Mantel--Debreu theorem \citep*{sonnenschein1972market,mantel1974characterization,debreu1974excess} implies that, in finite economies, aggregate excess demand can display essentially arbitrary local behavior around equilibrium prices. Stability therefore cannot be expected without additional structure. Despite these negative results, many sufficient conditions for stability operate through a common logic: the law of demand holds locally at equilibrium. More specifically, at an interior equilibrium, the Jacobian of aggregate excess demand can be decomposed into a negative semi-definite substitution component and a potentially destabilizing income effect component. Stability therefore follows whenever income effects are too weak to overturn substitution effects.

The purpose of this paper is to show that a large effective number of commodities can itself be a source of equilibrium stability through the accumulation of substitution effects across commodity dimensions. Such commodity dimensions may arise naturally in dynamic environments or in models with differentiated commodities. I show that high commodity dimensionality can help general equilibrium models exhibit well-behaved equilibrium properties even when goods are not gross substitutes, and without imposing strong curvature bounds on utility functions or conditions that generate approximate quasilinearity, where each agent’s shadow value of money is nearly constant.

We establish our results in dated-commodity exchange economies by fixing an underlying countably infinite-horizon exchange economy and studying its truncations at date $N$, where agents' utilities take the form $U_i(x_i)=\sum_{n=0}^N \beta^n u_{in}(x_{in})$. Stability does not arise merely from a large number of commodities: the additional dated commodities must also remain economically relevant. This is where patience matters. We therefore define the effective number of commodities as $N_\beta=\sum_{n=1}^N \beta^n$, the discounted count of commodities. This index is large only when the horizon is long and consumers are patient. Our main result shows that, under some baseline assumptions and a preference diversification assumption, once $N_\beta$ is sufficiently large, every equilibrium is locally tâtonnement stable. Consequently, every equilibrium has index $+1$, and equilibrium uniqueness follows from the index theorem.

The main difficulty in exploiting a large number of commodities is that commodity expansion cuts in both directions. As the number of commodities grows, expenditure is spread across a wider range of goods, so the income effect on any individual good is diluted. At the same time, however, the substitution effect between any two goods also becomes small. To make matters worse, in our setting, the largest nonzero eigenvalue of the substitution matrix can also converge to zero, reflecting the discounted utility structure. It is therefore not enough to argue that the income effect of each commodity, or even the operator norm of the income effect matrix, vanishes in large commodity spaces.

Instead, we show that, under an appropriate normalization, the equilibrium substitution effect accumulates at the rate of the effective number of commodities. The intuition is that, although the substitution effect between any two goods becomes small, the number of substitution opportunities increases with the number of commodities. These effects therefore add up across commodities, so the aggregate substitution effect grows at the rate of $N_\beta$.

Although a large effective number of commodities strengthens substitution through accumulation, it does not by itself guarantee stability. To obtain the required asymptotic comparison between substitution and income effects, we employ a preference diversification assumption. It requires that, as the number of commodities grows, agents' marginal expenditure-share profiles in equilibrium do not become concentrated along a small number of common directions. This limits the rate at which the income effect can accumulate. The condition is especially plausible when growth in the number of commodities reflects product differentiation: agents may have sparse or localized taste differences across commodities, as in Hotelling-type environments, or their taste profiles may become sufficiently heterogeneous as the commodity space expands. In this respect, it is in the spirit of the classical literature that uses dispersion conditions to control income effects, including \citet*{grandmont1992transformations}, \citet*{quah1997law}, and \citet*{jerison1999dispersed}. Nevertheless, because substitution effects themselves accumulate in our setting, the required diversification is substantially milder.

This paper shows that a large number of commodities provides another route to equilibrium stability and uniqueness in general equilibrium models. Classical sufficient conditions for stability include gross substitutability, as in \citet*{arrow1959stability}; quasi-linearity, for instance as a consequence of the permanent income hypothesis, as in \citet*{bewley1980permanent1}; and curvature bounds on utility that ensure the law of demand, beginning with \citet*{milleron1974unicite} and \citet*{polterovich1978criteria}.

The main stability mechanism in our paper is that a large number of commodities generates a strong accumulation of substitution effects. In this respect, our discussion is close to \citet*{quah2000monotonicity,quah2004aggregate} and \citet*{dana1995extension}, who employ curvature restrictions on utility to produce dominant substitution effects, a terminology used in \citet*{mas1991uniqueness}. In these works, when agents are heterogeneous, such curvature conditions are usually paired with distributional restrictions on income or preferences in order to obtain the law of demand, an approach going back at least to \citet*{hildenbrand1983law}. By contrast, in our setting, the large number of commodities itself becomes a source of dominant substitution effects and thereby weakens the distributional conditions needed to obtain stability.

Our discussion complements the long tradition of using a large number of commodities to control income effects. 
\citet*{vives1987small} and \citet*{hayashi2008note} show that income effects become small when the number of commodities is large, under restrictions that keep prices uniformly bounded away from zero.  
Such restrictions do not directly apply to our general equilibrium setting, in which equilibrium prices typically decline along the horizon at the rate of $\beta^n$. Relatedly, \citet*{keenan2025global} argue that a large number of commodities can make demand diffuse, thereby controlling income effects. Their maintained lower-bound assumption on the aggregate substitution effect is not directly suited to our discounted dated-commodity setting, where the largest nonzero eigenvalue of the equilibrium substitution matrix typically converges to zero as the horizon grows. Without imposing price lower bounds, \citet*{weretka2018quasilinear} and \citet*{weretka2026welfare}
use patience to obtain approximate quasi-linearity: finite saving changes are spread over many future periods, making marginal utilities of money nearly invariant. Our approach allows shadow values of money to vary, and even to fail to converge. Stability comes instead from the accumulation of aggregate substitution effects.

The size of the commodity space can be central for economic analysis around stability and uniqueness. At the
low-dimensional end, \citet*{GeanakoplosWalsh2018} study uniqueness and
stability in two-good economies. At the opposite end,
\citet*{hayashi2013smallness} works with a continuum of commodities and
shows that the income effect of an arbitrarily small commodity vanishes in
the limit. Countable-commodity economies have also been studied: \citet{kehoe1991gross} formulate excess
demand on infinite, potentially unbounded, price sequences and obtain uniqueness results
under gross substitutability. Very recently,
\citet*{lorenzoniwerning2025tatonnement} revisit tâtonnement stability in a
dynamic general equilibrium model with explicit price setting. See
\citet*{toda2024recent} for a recent survey.

The rest of the paper is organized as follows. Section \ref{section:model} describes the economies we study and defines stability. Section \ref{section:assumptions} states the assumptions. Section \ref{section:mainresult} presents the main results. Section \ref{section:proof} proves the main theorem by showing that the substitution effect accumulates while the income effect remains controlled as the number of commodities grows. Section \ref{section:discussion} discusses our modeling choices and assumptions, provides examples, and compares our stability condition with classical conditions. Section \ref{sec:conclude} concludes.

\section{Model}\label{section:model}

\subsection{Economy}

We begin with an underlying infinite-horizon environment with countably many agents and countably many commodities, and then define the finite economies we study as truncations of this environment.

Agents are indexed by the natural numbers $\mathbb{N}$, and commodities are indexed by the nonnegative integers $0,1,\dots$. We interpret this as a dated-commodity environment, where commodity $n$ is the date-$n$ commodity. Date $0$ is today, and all other dates are future dates. We restrict attention to additively separable preferences. For each agent $i\in\mathbb N$ and date $n\ge 0$, let $u_{in}:\mathbb R_+\to\mathbb R$ denote agent $i$'s period-$n$ utility function, and let $\omega_{in}$ denote her date-$n$ endowment. We assume that each infinite-horizon endowment vector $(\omega_{i0},\omega_{i1},\dots)$ belongs to $\ell^\infty(\mathbb R_+)$.

For any discount factor $\beta\in(0,1)$, agent $i$'s utility over consumption streams $x_i=(x_{i0},x_{i1},\dots)\in \ell^\infty(\mathbb R_+)$ is
$$
U_i^\beta(x_i)=\sum_{n=0}^\infty \beta^n u_{in}(x_{in}),
$$
where $\beta$ is a common discount factor.\footnote{Section \ref{subsec:choiceofutilityform} discusses the modeling role of discounting.}

Our main interest is in regimes with both a large number of commodities and sufficiently patient consumers. For this reason, we index each finite economy by its terminal date $N\ge 1$ and the discount factor $\beta\in(0,1)$.

For each pair $(N,\beta)$, let $\mathcal E^{N,\beta}=\bigl(I_N,(U_i^{N,\beta},\omega_i^N)_{i\in I_N}\bigr)$
denote the corresponding finite economy, in which commodities are indexed by dates $0,1,\dots,N$. We let $I_N=\{1,2,\dots,I_N\}$ denote the set of agents in this truncated economy, and we use $I_N$ also to denote its cardinality. As $N$ varies, we allow the population size $I_N$ to vary with $N$, though this is not required. Agent $i$'s preferences on $\mathbb R_+^{N+1}$ are represented by
$$
U_i^{N,\beta}(x_i)=\sum_{n=0}^N \beta^n u_{in}(x_{in}),
$$
where $x_i=(x_{i0},x_{i1},\dots,x_{iN})\in\mathbb R_+^{N+1}$, and her truncated endowment vector is
$\omega_i^N=(\omega_{i0},\omega_{i1},\dots,\omega_{iN})\in\mathbb R_+^{N+1}$.

For each date $n=0,\dots,N$, aggregate endowment of commodity $n$ in economy $\mathcal E^{N,\beta}$ is denoted by
$$
\Omega_n^N=\sum_{i\in I_N}\omega_{in}.
$$

We define the effective number of commodities in $\mathcal E^{N,\beta}$ by
$$
N_\beta:=\sum_{n=1}^N \beta^n.
$$
This is the discounted count of future commodities in the economy. Our asymptotic analysis is carried out in regimes in which $N_\beta$ is large. Thus, what matters is not only that the number of commodities $N$ be large, but also that consumers be sufficiently patient, so that distant commodities continue to matter for agents' decisions. In particular, if $N_\beta$ becomes arbitrarily large, then necessarily $N$ becomes arbitrarily large and $\beta$ becomes arbitrarily close to one.

\subsection{Equilibrium and Stability}

We now define competitive equilibrium and local stability for a fixed economy $\mathcal E^{N,\beta}$. Throughout the paper, we suppress the superscript $(N,\beta)$ whenever no confusion can arise. In principle, all objects depend on the underlying economy under consideration.

We first define competitive equilibrium. Throughout, we normalize the price of commodity $0$ to one, and let $p=(p_1,\dots,p_N)\in\mathbb R_{++}^N$ denote the prices of future commodities. At price $p$, agent $i$ chooses a consumption bundle $x_i\in \mathbb{R}_+^{N+1}$ to maximize $U_i$ subject to the budget constraint $x_{i0}+\sum_{n=1}^N p_n x_{in}\le \omega_{i0}+\sum_{n=1}^N p_n\omega_{in}.$
Let $\xi_i(p)=(\xi_{i0}(p),\dots,\xi_{iN}(p))$ denote agent $i$'s Marshallian demand, the maximizer of agent $i$'s problem, in $\mathcal E^{N,\beta}$. Aggregate excess demand for future commodities is denoted by $z(p)=(z_1(p),\dots,z_N(p))$, where
$$
z_n(p):=\sum_{i\in I_N}\xi_{in}(p)-\Omega_n^N,
\qquad n=1,\dots,N.
$$

A competitive equilibrium of $\mathcal E^{N,\beta}$ is a pair $(\bar p,\bar x)$ such that $\bar p\in\mathbb R_+^N$ satisfies market clearing for all future commodities,
$z(\bar p)=0,$ and the allocation $\bar x=(\bar x_i)_{i\in I_N}$ satisfies $\bar x_i=\xi_i(\bar p)$ for every $i\in I_N$. The market clearing condition for date-$0$ market is implied by Walras' law.

We study the standard continuous-time tâtonnement process
$$
p'(t)=z(p(t)), \qquad t\ge 0.
$$
This process formalizes the idea that prices adjust in the direction of excess demand.

We say that an interior equilibrium $(\bar p,\bar x)$ is \emph{locally tâtonnement stable} if there exists an open neighborhood $P$ of $\bar p$ such that, for every initial condition $p(0)\in P$, the solution $p(t)$ converges to $\bar p$. By the standard linearization principle for smooth differential equations, local tâtonnement stability is implied if $Dz(\bar p)$ is negative definite, where
$Dz(\bar p)=\left(\frac{\partial z_n}{\partial p_m}(\bar p)\right)_{m,n=1}^N$
is the Jacobian matrix of aggregate excess demand at $\bar p$.

\section{Assumptions}\label{section:assumptions}

We impose five assumptions. The first four are baseline assumptions on endowments and preferences and the fifth is a preference diversification assumption. As the effective number of commodities grows, the former ensures that substitution effects accumulate, while the latter keeps the income effect under control.

To state the assumptions, for each agent $i\in \mathbb{N}$, we define two quantities. First, for each date $n\ge 0$, let agent $i$'s date-$n$ risk-tolerance function be
$$
r_{in}(x):=\frac{u_{in}'(x)}{-u_{in}''(x)}.
$$
This is the reciprocal of the Arrow-Pratt coefficient of absolute risk aversion for period-$n$ utility, and will be useful in characterizing how agent $i$'s date-$n$ demand responds to wealth changes. In addition, for each economy $\mathcal{E}^{N,\beta}$ and each dated-commodity $1\le n\le N$, we define its normalized commodity weight as 
$$
\pi_n^{N,\beta}:=\frac{\beta^n}{N_\beta}.
$$
The coefficient $\pi_n^{N,\beta}$ is date $n$'s share in the effective number of commodities $N_\beta$.

\subsection{Baseline Assumptions}

\begin{enumerate}
	\item[(A1)] For every $i\in \mathbb{N}$ and every $n\ge 0$, $u_{in}$ is twice continuously differentiable on $\mathbb R_{++}$, satisfies $u_{in}'(x)>0$ and $u_{in}''(x)<0$ for every $x>0$, and satisfies the Inada condition $\lim_{x\to 0}u_{in}'(x)=+\infty.$
	
	\item[(A2)] There exist constants $0<c_u<C_u<\infty$ such that, for every $i\in\mathbb N$, every $n\ge 0$, and every $x>0$,
	$$
	c_u \le \frac{u_{in}'(x)}{u_{i0}'(x)} \le C_u, \qquad \text{~and~}
	\qquad 
	c_u x \le r_{in}(x)\le C_u x.
	$$

	\item[(A3)] There exist constants $0<c_\Omega < C_\Omega<\infty$ such that for every  $N\ge 1$ and $0\le n\le N$, 
		$$0<c_\Omega \le \frac{\Omega_n^N}{I_N}\le C_\Omega <\infty.$$
	
\item[(A4)] There exist constants $0<c_W<C_W<\infty$ such that, for every $N\ge 1$ and every $i\in I_N$,
$$
\sum_{n=1}^N \beta^n \omega_{in} \ge c_W N_\beta, \qquad \text{~and~}
\qquad \omega_{in}\le C_W  \qquad \text{~for~every~} n=0,\dots,N.
$$
\end{enumerate}

Assumptions \textnormal{(A1)} and \textnormal{(A2)} impose restrictions on preferences. Assumption \textnormal{(A1)} collects standard regularity conditions. Assumption \textnormal{(A2)} imposes uniform bounds on marginal utility and risk tolerance. Both assumptions are satisfied by isoelastic utility functions with taste variations, where
$$
u_{in}(x)=\tau_{in}
\begin{cases}
	\dfrac{x^{1-1/\sigma_{i}}-1}{1-1/\sigma_{i}}, & \sigma_{i}>0,\ \sigma_{i}\neq 1,\\
	\log x, & \sigma_{i}=1,
\end{cases}
$$
and $\tau_{in}>0$ is interpreted as agent $i$'s taste on commodity $n$. In this case, $u_{in}'(x)=\tau_{in}x^{-1/\sigma_i}$, $u_{in}''(x)= - \frac{\tau_{in}}{\sigma_{i}}x^{-1/\sigma_{i}-1}$ and $r_{in}(x)=\sigma_{i}x$. Therefore, Assumptions \textnormal{(A1)} and \textnormal{(A2)} hold whenever the relevant taste	coefficients are uniformly bounded: $0<\underline{\tau}\le \tau_{in}\le \bar{\tau}<\infty$ for all $i\in\mathbb{N}$ and $n\ge 0$.

Assumptions \textnormal{(A3)} and \textnormal{(A4)} impose restrictions on endowments. Assumption \textnormal{(A3)} requires per-capita aggregate endowment at each date to remain uniformly bounded from above and away from zero. In particular, it rules out situations in which some dates become asymptotically negligible or asymptotically dominant. Assumption \textnormal{(A4)} requires each agent's total discounted endowment to be comparable to the effective number of commodities $N_\beta$ and individual endowments to remain bounded as the number of commodities grows. This rules out cases in which, for some agents, almost all relevant endowment is concentrated in only a few commodities or individual endowments become unbounded. Both assumptions are satisfied under a stronger interior-endowment assumption: there exist constants $0<c\le C<\infty$ such that
$c\le \omega_{in}\le C$ holds for all $i\in \mathbb{N}$ and $n\ge 0$.

\subsection{Preference Diversification Assumption}\label{subsec:diversification-assumption}

To control the income effect as the effective number of commodities grows, we impose a preference diversification assumption. The relevant object is the vector of marginal expenditure shares induced by a small change in future wealth.

For each agent $i$ in economy $\mathcal{E}^{N,\beta}$ and each future consumption bundle $x_i=(x_{i1},\dots,x_{iN}) \in\mathbb{R}_{++}^N$, define
$$
m_{in}(x_i):=
\frac{\beta^n u_{in}'(x_{in})\,r_{in}(x_{in})}
{\sum_{m=1}^N \beta^m u_{im}'(x_{im})\,r_{im}(x_{im})},
\qquad n=1,\dots,N.\footnote{Under \textnormal{(A1)} and \textnormal{(A2)}, we have $m_{in}(x_i)>0$ for every interior future consumption vector, so this quantity is well defined.}
$$
In \ref{appendix:marginal-spending-profile}, we show that $m_{in}(x_i)$ can be interpreted as agent $i$'s marginal expenditure share on commodity $n$ out of an additional unit of future wealth. The reason is that, when $x_i$ is an optimal consumption choice under price $p$, $m_{in}(x_i)=p_n\frac{\partial x_{in}}{\partial w_i}$, where $w_i$ denotes wealth spent on future commodities. Equivalently, it is the price-weighted slope of the Engel curve of commodity $n$ evaluated at the relevant future wealth.

For each allocation $x=(x_i)_{i\in I_N}$ with $x_i\in \mathbb{R}_{++}^N$ for all $i\in I_N$, let
$$
\bar m_n(x)=\frac{1}{I_N}\sum_{i\in I_N} m_{in}(x_i)
$$
denote the population average of marginal expenditure shares on commodity $n$. Define the relative deviation of agent $i$'s marginal expenditure share on commodity $n$ from this population average by
\begin{equation}
	\rho_{in}(x)=\frac{m_{in}(x_i)-\bar m_n(x)}{\bar m_n(x)}.
	\label{eqn:relativedeviation}
\end{equation}

Let $\rho_i(x)=(\rho_{i1}(x),\dots,\rho_{iN}(x))$, and write $\bar\rho_i=\rho_i(\bar x)$ for an equilibrium allocation $\bar x$. The vector $\bar\rho_i$ captures agent $i$'s normalized expenditure-share profile, where normalization is by relative deviation from the population average. When we wish to emphasize the underlying economy, we write $\bar\rho_i^{N,\beta}=\rho_i(\bar x^{N,\beta})$
for an equilibrium allocation $\bar x^{N,\beta}$ of economy $\mathcal E^{N,\beta}$.

We impose an assumption to prevent these normalized expenditure-share profiles from staying strongly aligned across agents at equilibrium. To measure such alignment, for any two vectors in $\mathbb{R}^N$, define
\begin{equation}
	\langle x,y\rangle_{N,\beta}=\sum_{n=1}^N \pi_n^{N,\beta} x_ny_n=\frac{1}{N_\beta}\sum_{n=1}^N \beta^n x_n y_n.
	\label{eqn:n-product}
\end{equation}
This inner product measures weighted average alignment across dates, where the weights $\pi_n^{N,\beta}$ reflect the economic relevance of date-$n$ commodity. We denote the induced norm by
\begin{equation}
	\norm{x}_{N,\beta}=\langle x,x\rangle_{N,\beta}^{1/2}=\left(\sum_{n=1}^N \pi_n^{N,\beta} x_n^2\right)^{1/2}.
	\label{eqn:n-norm}
\end{equation}

We now state the preference diversification assumption:
\begin{enumerate}
	\item[(A5)] For any family of equilibria $\{(\bar p^{N,\beta},\bar x^{N,\beta})\}$ such that $(\bar p^{N,\beta},\bar x^{N,\beta})$ is an equilibrium of $\mathcal E^{N,\beta}$ for each $(N,\beta)$,
	$$
	\frac{1}{I_N^2}\sum_{i,j\in I_N}
	\left|\left\langle \bar\rho_i^{N,\beta},\bar\rho_j^{N,\beta}\right\rangle_{N,\beta}\right|
	\to 0
	\qquad\text{as } N_\beta\to\infty.
	$$
\end{enumerate}

Assumption \textnormal{(A5)} is an equilibrium-level preference diversification condition. It requires that, along any family of equilibria, the average weighted alignment of agents' normalized marginal expenditure-share profiles vanishes as the effective number of commodities grows. Equivalently, it rules out the case in which many agents' relative deviations from the population average remain concentrated along only a small number of common directions. In this sense, the condition is a local, high-dimensional analogue of the idea that dispersion in individual demand behavior can regularize aggregate demand.\footnote{This perspective is related to the classical aggregation literature on the law of demand and the weak axiom, including \citet*{hildenbrand1983law}, \citet*{grandmont1987distributions,grandmont1992transformations}, \citet*{freixas1987engel}, \citet*{jerison1999dispersed}, and \citet*{quah1997law}. Unlike the conditions in these papers, the present condition is local and is used only to control the income-effect term in the large-$N_\beta$ limit.}

Assumption \textnormal{(A5)} is tractable in structured environments. For instance, in the logarithmic case $u_{in}(x)=\tau_{in}\log x$, the marginal expenditure-share profile $m_i(x_i)$ is independent of $x_i$. Thus, the condition becomes a primitive restriction on taste-deviation profiles and admits a no-concentration interpretation, as Section \ref{subsection:assumptiondiscussion} shows. We also give a non-logarithmic example, with the aid of additional structure on endowments, in Section \ref{subsec:comparisons}. 

This preference diversification condition may arise naturally in two scenarios: when taste differences are sparse or localized across commodities, as in Hotelling-type environments, or when taste profiles become sufficiently heterogeneous as the commodity space grows; Section \ref{subsection:examples} discusses both cases. The former scenario is compatible with a fixed population size, $I_N=I$, because each agent's normalized deviation profile can become small in the weighted commodity space. The latter scenario is more plausible in economies with a large number of agents: when the number of commodities is large, a sufficiently rich class of agents can make the average pairwise alignment of taste-deviation profiles small. Our formulation allows such economies.

Finally, Assumption \textnormal{(A5)} can be weakened formally: as the proof of Proposition \ref{prop:M-upper-bound} shows, Theorem \ref{thm:main} and Corollary \ref{cor:global-uniqueness} continue to hold if \textnormal{(A5)} is replaced by the following weaker condition:
\begin{enumerate}
	\item[(A5')] For any family of equilibria $\{(\bar p^{N,\beta},\bar x^{N,\beta})\}$, where $(\bar p^{N,\beta},\bar x^{N,\beta})$ is an equilibrium of $\mathcal E^{N,\beta}$ for each $(N,\beta)$,
	$$
	\frac{1}{I_N^2}\sum_{i,j\in I_N}
	\langle t_i^{N,\beta},t_j^{N,\beta}\rangle_{N,\beta}
	\langle \bar\rho_i^{N,\beta},\bar\rho_j^{N,\beta}\rangle_{N,\beta}
	\to 0
	\qquad\text{as } N_\beta\to\infty,
	$$
	where $t_i^{N,\beta}=(\bar x_{in}^{N,\beta}-\omega_{in})_{n=1}^N$ is agent $i$'s equilibrium net-trade profile.
\end{enumerate}
Condition \textnormal{(A5')} controls the cross-alignment between equilibrium net-trade profiles and equilibrium marginal-expenditure-share deviation profiles. It is implied by \textnormal{(A5)}, because equilibrium net trades are uniformly bounded under the baseline assumptions, but it is considerably weaker. We state and discuss Assumption \textnormal{(A5)} because it has a more transparent interpretation as a preference-diversification condition. Section \ref{subsection:necessity} shows why some assumption of this kind is needed for our analysis: in a two-type economy where preferences fail to diversify as the number of commodities becomes large, the income effect can grow at the same rate as the substitution effect.

\section{Equilibrium Stability and Uniqueness}\label{section:mainresult}

Our main result states that every competitive equilibrium is locally tâtonnement stable when the effective number of commodities $N_\beta$ is sufficiently large. Indeed, we prove a slightly stronger result: the Jacobian of aggregate excess demand at any equilibrium price is negative definite.

\begin{theorem}\label{thm:main}
	Suppose Assumptions \textnormal{(A1)}--\textnormal{(A5)} hold. Then there exists a constant $N^*$ such that, whenever $N_\beta>N^*$, every equilibrium price $\bar p$ of $\mathcal E^{N,\beta}$ is locally tâtonnement stable. Moreover, $Dz(\bar p)$ is negative definite.
\end{theorem}

Theorem \ref{thm:main} has an immediate implication for equilibrium uniqueness: whenever the Jacobian of aggregate excess demand at an equilibrium price $\bar{p}$ is negative definite, the equilibrium is regular and has index $+1$. Hence, equilibrium uniqueness follows as a corollary.

\begin{corollary}\label{cor:global-uniqueness}
	Suppose Assumptions \textnormal{(A1)}--\textnormal{(A5)} hold. If $N_\beta>N^*$, where $N^*$ is as in Theorem \ref{thm:main}, then $\mathcal E^{N,\beta}$ admits a unique equilibrium.
\end{corollary}
\begin{proof}
	Fix an equilibrium price vector $\bar p$ of $\mathcal E^{N,\beta}$. Theorem \ref{thm:main} states that $Dz(\bar p)$ is negative definite whenever $N_\beta>N^*$. Thus, $\bar p$ is regular and has index $+1$, where the index at $\bar{p}$ is defined by $\operatorname{ind}(\bar p)=\operatorname{sgn}\det[-D z(\bar p)]$. Since $\bar p$ was arbitrary, every equilibrium has index $+1$. By the index theorem, the sum of equilibrium indices is $+1$; see \citet{dierker1972number}. Therefore, $\mathcal E^{N,\beta}$ admits a unique equilibrium price vector.
\end{proof}

\section{Proof of Main Theorem}\label{section:proof}

In this section, we prove the main theorem by showing that, at any equilibrium price $\bar p$, the Jacobian $Dz(\bar p)$ is negative definite whenever the effective number of commodities is sufficiently large. The argument proceeds in three steps. First, we show that the baseline assumptions imply uniform bounds on equilibrium objects. Second, we show that substitution effects naturally accumulate at rate $N_\beta$. Third, we use Assumption \textnormal{(A5)} to show that the income effect grows at a lower order than $N_\beta$. As a consequence, when the effective number of commodities is sufficiently large, the substitution effect dominates the income effect, which yields stability.

As in the model description, we suppress the superscript $(N,\beta)$ whenever no confusion can arise. Throughout the proof, unless otherwise specified, we consider an arbitrary competitive equilibrium $(\bar p,\bar x)$ of a generic economy $\mathcal E^{N,\beta}$, where $N\in\mathbb N$ and $\beta\in(0,1)$. All constants appearing in the bounds below are independent of the agent $i$, the commodity $n$, the patience parameter $\beta$, the truncation horizon $N$, and the equilibrium under consideration.

\subsection{Slutsky Decomposition}

We fix $N\ge 1$ and $\beta\in(0,1)$, and also fix an equilibrium price $\bar p$ of the truncated economy $\mathcal E^{N,\beta}$.  First, we observe that every equilibrium price $\bar p$ must be interior.

\begin{lemma}\label{lem:equilibrium-interior}
	Suppose Assumption \textnormal{(A1)} holds. Then every equilibrium price of $\mathcal E^{N,\beta}$ is interior.
\end{lemma}
\begin{proof}
	See \ref{proof:interiorequilibrium}.
\end{proof}

The next lemma rewrites each individual's Slutsky decomposition using risk tolerances.

\begin{lemma}\label{lem:substitution-income-decomposition}
	Suppose Assumption \textnormal{(A1)} holds and $(\bar p,\bar x)$ is a competitive equilibrium of $\mathcal E^{N,\beta}$. For each agent $i\in I_N$ and each $m,n=1,\dots,N$,
	\begin{equation}\label{eq:exchange-slutsky}
		\frac{d\xi_{in}}{dp_m}(\bar p)= S_i(m,n)+ M_i(m,n),
	\end{equation}
	where $S_i$ is the substitution matrix and $M_i$ is the corresponding income-effect matrix, defined by
	$$
	S_i(m,n):=\frac{\bar r_{im}\bar r_{in}}{\bar r_i}-\frac{\bar r_{in}}{\bar p_n}\mathbf 1_{\{m=n\}},
	\qquad
	M_i(m,n):=\frac{\bar r_{in}}{\bar r_i}(\omega_{im}-\bar x_{im}).
	$$
	Here, $	\bar r_{in}:=\frac{u_{in}'(\bar x_{in})}{-u_{in}''(\bar x_{in})}$ is agent $i$'s equilibrium risk tolerance for commodity $n$, and $\bar r_i:=\bar r_{i0}+\sum_{n=1}^N \bar p_n \bar r_{in}$ is the price-weighted sum of agent $i$'s equilibrium risk tolerances.
	\end{lemma}
\begin{proof}
	See \ref{proof:sluskydecom}.
\end{proof}

Applying Lemma \ref{lem:substitution-income-decomposition}, we can express the individual substitution and income effects under price change  $q$ as follows.

\begin{lemma}\label{lem:substitutionandincomeeffect}
	For each $i\in I_N$ and any price perturbation $q\in\mathbb R^N$, we have
	\begin{equation}
		q^T S_i q
		=
		-\sum_{n=1}^N \frac{\bar r_{in}}{\bar p_n}\bigl(q_n-\bar p_n\Lambda_i(q)\bigr)^2
		-
		\frac{\bar r_{i0}\bar r_i^0}{\bar r_i}\Lambda_i(q)^2,
		\label{eq:substitution-quadratic}
	\end{equation}
	and
	\begin{equation}
		q^T M_i q
		=
		\frac{\bar r_i^0}{\bar r_i}\Lambda_i(q)\sum_{m=1}^N (\omega_{im}-\bar x_{im})q_m,
		\label{eq:income-quadratic}
	\end{equation}
	where
	\begin{equation}
		\bar r_i^0:=\sum_{n=1}^N \bar p_n \bar r_{in},
		\qquad
		\bar{m}_{in}:=\frac{\bar p_n \bar r_{in}}{\bar r_i^0},
		\qquad
		\Lambda_i(q):=\frac{\sum_{n=1}^N \bar r_{in}q_n}{\bar r_i^0}
		=
		\sum_{n=1}^N \bar{m}_{in}\frac{q_n}{\bar p_n}.
		\label{eqn:barm_in}
	\end{equation}
\end{lemma}

\begin{proof}
	See \ref{proof:substitutionincomeeffect}.
\end{proof}

Here, $\bar r_i^0$ is agent $i$'s price-weighted sum of future equilibrium risk tolerances, and $\bar m_{in}=m_{in}(\bar{x}_{i})$ is the equilibrium expenditure share of an additional unit of future wealth allocated to commodity $n$; see \ref{appendix:marginal-spending-profile}. In particular, $(\bar m_{in})_{n=1}^N$ is a probability vector. The term $\Lambda_i(q)\bar p$ is the component of $q$ proportional to the equilibrium price vector from agent $i$'s perspective, while $q-\Lambda_i(q)\bar p$ is the corresponding price distortion.

We now study the Jacobian of aggregate excess demand at $\bar p$. It is well-known that, due to the compensated law of demand, substitution effects respond differently to proportional price changes and price distortions. For this reason, we decompose $q=\alpha \bar p+u$, where $\alpha \bar p$ is a proportional price change and $u$ is the remaining price distortion. Furthermore, we choose this decomposition such that $\Psi(u)=0$, where the linear functional $\Psi:\mathbb R^N\to\mathbb R$ is defined by
$$
\Psi(q):=\sum_{i\in I_N}\frac{2\bar r_{i0}+\omega_{i0}-\bar x_{i0}}{\bar r_i}\,\bar r_i^0\Lambda_i(q).
$$
This normalization is chosen to simplify the decomposition of  $q^T Dz(\bar{p}) q$. The following proposition justifies the role of this decomposition. 

\begin{proposition}\label{prop:quadratic-decomposition}
	Suppose $q=\alpha\bar p+u$, where $\alpha\in\mathbb R$ and $u$ satisfies $\Psi(u)=0$. Then
	\begin{equation}
		q^T Dz(\bar{p}) q=-A \alpha^2+ R(u) \alpha-S(u)+M(u),
		\label{eqn:AISdecomposition}
	\end{equation}
	where
	\begin{equation}
		A:=\sum_{i\in I_N}\frac{\bar r_i^0}{\bar r_i}\bigl(\bar r_{i0}+\omega_{i0}-\bar x_{i0}\bigr),
		\label{eqn:A}
	\end{equation}
	\begin{equation}
		S(u):=\sum_{i\in I_N}\sum_{n=1}^N \frac{\bar r_{in}}{\bar p_n}\bigl(u_n-\bar p_n\Lambda_i(u)\bigr)^2,
		\label{eqn:S}
	\end{equation}
	\begin{equation}
		R(u):=
		\sum_{i\in I_N} \frac{\bar r_{i0}}{\bar r_i}\left(\sum_{m=1}^N
		(\bar x_{im}-\omega_{im})u_m
		\right),
		\label{eqn:H}
	\end{equation}
	\begin{equation}
		M(u):=\sum_{i\in I_N}\frac{\bar r_i^0}{\bar r_i}\Lambda_i(u)
		\left[
		\sum_{m=1}^N (\omega_{im}-\bar x_{im})u_m-\bar r_{i0}\Lambda_i(u)
		\right].
		\label{eqn:M}
	\end{equation}
	Moreover, if $A>0$, then $\Psi(\bar p)>0$. Consequently, every price change $q\in\mathbb R^N$ admits a unique decomposition of the form $q=\alpha\bar p+u$ with $\Psi(u)=0$.
\end{proposition}
\begin{proof}
	See \ref{proof:quadratic-decomposition}.
\end{proof}

Here, $S(u)$ captures the main part of the substitution effect generated by the price distortion component $u$, $M(u)$ captures the main part of the income effect generated by $u$, $-A\alpha^2$ captures the substitution effect along the proportional price change component,\footnote{Even when future prices change proportionally, good $0$ remains the numéraire with price normalized to $1$. Thus, the overall price change is not fully proportional across all goods, so a residual substitution effect remains.} and $R(u)\alpha$ is the remaining mixed term.

\subsection{Equilibrium Bounds}

We next state the equilibrium bounds implied by Assumptions \textnormal{(A1)}--\textnormal{(A4)}. These bounds will be used repeatedly in the following proof. In particular, we show that equilibrium future prices are uniformly comparable to the discount profile $(\beta^n)_{n=1}^N$, while equilibrium consumption and risk tolerances remain uniformly bounded.

\begin{proposition}\label{prop:equilibrium-bounds}
	Under \emph{(A1)}--\emph{(A4)}, there exist positive constants
	$c_p,C_p,c_x,C_x,c_r,C_r$ and a threshold $N_0\in\mathbb{N}$, independent of $i$, $n$, $N$, $\beta$, and the equilibrium, such that whenever $N_\beta\ge N_0$, at every equilibrium $(\bar p,\bar x)$ of $\mathcal E^{N,\beta}$:
	\begin{enumerate}
		\item[(i)] For every date $n=1,\dots,N$,
		$$
		c_p\,\beta^n \le \bar p_n \le C_p\,\beta^n.
		$$
		\item[(ii)] For every agent $i\in I_N$ and every date $n=0,1,\dots,N$,
		$$
		c_x \le \bar x_{in} \le C_x
		\qquad\text{and}\qquad
		c_r \le \bar r_{in} \le C_r.
		$$
	\end{enumerate}
\end{proposition}

\begin{proof}
	See \ref{proof:equilibrium-bounds}.
\end{proof}

An immediate consequence of this proposition is that both $\bar r_i^0$ and $\bar r_i$ are uniformly comparable to the effective number of commodities $N_\beta$. Formally, for every $i\in I_N$,
\begin{equation}
	c_p c_r N_\beta \le \bar r_i^0 \le C_p C_r N_\beta,
	\qquad
	c_p c_r N_\beta \le \bar r_i \le C_r+C_p C_r N_\beta.
	\label{eq:derived-risk-bounds}
\end{equation}

The next two lemmas record a consequence of the equilibrium bounds that ensures the decomposition in Proposition \ref{prop:quadratic-decomposition} is well-defined and the mixed term $R(u)$ is bounded.

\begin{lemma}\label{lem:A-positive}
	Under \emph{(A1)}--\emph{(A4)}, there exist constants $c_A>0$ and $N_1 \ge N_0$ such that
	\begin{equation}
		A\ge c_A I_N,
		\label{eq:A-positive}
	\end{equation}
	whenever $N_\beta\ge N_1$. Here, $N_0$ is the constant specified in Proposition \ref{prop:equilibrium-bounds}. 
\end{lemma}

\begin{proof}
	See \ref{proof:A-positive}.
\end{proof}

\begin{lemma}\label{lem:R-order-one}
	Under \emph{(A1)}--\emph{(A4)}, there exists a constant $C_R>0$ such that for every $u\in\mathbb R^N$, writing $v_n:=u_n/\bar p_n$, one has
	\begin{equation}
		|R(u)|\le C_R I_N \norm{v}_{N,\beta},
		\label{eq:R-order-one}
	\end{equation}
	provided $N_\beta\ge N_0$, where $N_0$ is the constant specified in Proposition \ref{prop:equilibrium-bounds}. 
\end{lemma}

\begin{proof}
	See \ref{proof:R-order-one}.
\end{proof}

\subsection{Accumulation of Substitution Effects}

In this subsection, we show that the distortion-driven substitution effect $S(u)$  accumulates at the order of $N_\beta$. 

The first step is to write $S(u)$ as a weighted pairwise-difference form. The weight $w_{mn}$ between $m$ and $n$ is large when many agents have large marginal expenditure shares on both commodities $m$ and $n$. In that case, a larger gap between normalized price distortion $v_m$ and $v_n$ contributes more to the substitution effect.

\begin{lemma}\label{lem:S-structure}
	For every $u\in\mathbb R^N$, write $v_n=\frac{u_n}{\bar p_n}$ for every $n=1,\dots,N$. We have
	\begin{equation}
		S(u)=\frac12\sum_{m,n=1}^N w_{mn}(v_m-v_n)^2,
		\label{eq:S-graph}
	\end{equation}
	where
	\begin{equation}
		w_{mn}:=\sum_{i\in I_N}\bar r_i^0\,\bar m_{im}\bar m_{in}
		=
		\sum_{i\in I_N}\frac{(\bar p_m\bar r_{im})(\bar p_n\bar r_{in})}{\bar r_i^0}.
		\label{eq:w-mn}
	\end{equation}
\end{lemma}

\begin{proof}
	Recall $\Lambda_i(u)=\sum_{n=1}^N \bar m_{in}v_n$ and $	\bar p_n\bar r_{in}=\bar r_i^0\bar m_{in}.$ Hence, agent $i$'s term in $S(u)$ is
	$$
	\sum_{n=1}^N \frac{\bar r_{in}}{\bar p_n}\bigl(u_n-\bar p_n\Lambda_i(u)\bigr)^2
	=
	\bar r_i^0\sum_{n=1}^N \bar m_{in}\left(v_n-\sum_{m=1}^N \bar m_{im}v_m\right)^2.
	$$
	Since $(\bar m_{in})_{n=1}^N$ is a probability vector, the variance identity yields
	$$
	\sum_{n=1}^N \bar m_{in}\left(v_n-\sum_{m=1}^N \bar m_{im}v_m\right)^2
	=
	\frac12\sum_{m,n=1}^N \bar m_{im}\bar m_{in}(v_m-v_n)^2.
	$$
	Summing over $i\in I_N$ proves \eqref{eq:S-graph}.
\end{proof}

The next lemma shows that these weights $w_{mn}$ uniformly dominate the
benchmark weight product $\pi_m^{N,\beta}\pi_n^{N,\beta}$. 

\begin{lemma}\label{lem:w-lower-bound}
	Under \emph{(A1)}--\emph{(A4)}, there exist constants $0<c_w<C_w<
	\infty$ such that for every $N\ge 1$ and every $m,n=1,\dots,N$,
	\begin{equation}
		c_w\,I_N N_\beta\,\pi_m^{N,\beta}\pi_n^{N,\beta} \le w_{mn}\le C_w \,I_N N_\beta\,\pi_m^{N,\beta}\pi_n^{N,\beta},
		\label{eq:w-lower-bound}
	\end{equation}
	provided $N_\beta\ge N_0$, where $N_0$ is the constant specified in Proposition \ref{prop:equilibrium-bounds}. 
\end{lemma}

\begin{proof}
	Using Proposition \ref{prop:equilibrium-bounds} and \eqref{eq:derived-risk-bounds}, we have $
	c_rc_p\beta^n\le \bar p_n\bar r_{in}\le C_rC_p\beta^n$ and $c_pc_r N_\beta\le \bar r_i^0\le C_pC_r N_\beta$. Therefore, using the definition of $w_{mn}$ in \eqref{eq:w-mn}, we have
	$$
	\frac{c_r^2c_p^2}{C_pC_r}\sum_{i\in I_N}\frac{\beta^m\beta^n}{N_\beta}\le w_{mn}
	\le
	\frac{C_r^2C_p^2}{c_pc_r}\sum_{i\in I_N}\frac{\beta^m\beta^n}{N_\beta}.
	$$
	Recall $\pi_n^{N,\beta}=\beta^n/N_\beta$ and take $c_w=\frac{c_r^2c_p^2}{C_pC_r}$ and $C_w=\frac{C_r^2C_p^2}{c_pc_r}$, the lemma follows.
\end{proof}

We now prove our key observation about the stabilizing force: as the effective number of commodities grows, the substitution effect $S(u)$ accumulates at the order $N_\beta$. Note that when $v=\mathbb{1}$, by Lemma \ref{lem:S-structure}, $S(u)=0$. Thus, the normalization that $u$ is a price distortion direction satisfying $\Psi(u)=0$ plays a role in the following proposition.

\begin{proposition}\label{prop:S-lower-bound}
	Suppose \emph{(A1)}--\emph{(A4)} hold. Then, there exists a constant $c_S>0$ and an integer $N_2\ge N_1$ such that for every decomposition $q=\alpha\bar p+u$ with $\Psi(u)=0$ and $v_n=u_n/\bar p_n$, one has
	\begin{equation}
		S(u)\ge c_S\,I_N N_\beta \norm{v}_{N,\beta}^2,
		\label{eq:S-lower-bound}
	\end{equation}
	whenever $N_\beta\ge N_2$. Here, $N_1$ is the constant specified in Lemma \ref{lem:A-positive}.
\end{proposition}

\begin{proof}
	For every $N\ge 1$ and $n=1,\dots,N$, define $d_n=\sum_{i\in I_N}\frac{2\bar r_{i0}+\omega_{i0}-\bar x_{i0}}{\bar r_i}\,\bar p_n\bar r_{in}.$
	Since $u_n=\bar p_n v_n$ and $u\in\ker\Psi$, we have
	\begin{equation}
		\sum_{n=1}^N d_n v_n=0.
		\label{eq:d-v-orthogonality}
	\end{equation}
	
	By Proposition \ref{prop:equilibrium-bounds}, the quantities $\bar r_{i0}$ and $\bar x_{i0}$ are uniformly bounded. By Assumption \textnormal{(A4)}, $\omega_{i0}$ is also uniformly bounded. Hence, there exists a constant $C>0$ such that
	$|2\bar r_{i0}+\omega_{i0}-\bar x_{i0}|\le C$ for all $i\in I_N,\ N\ge 1$. By Proposition \ref{prop:equilibrium-bounds} and \eqref{eq:derived-risk-bounds}, we obtain
	$$
	|d_n|
	\le
	C\sum_{i\in I_N}\frac{\bar p_n\bar r_{in}}{\bar r_i}
	\le
	C_d I_N\frac{\beta^n}{N_\beta}
	=
	C_d I_N \pi_n^{N,\beta}
	\qquad\text{for every } N\ge N_0, n=1,\dots, N,
	$$
	for some constant $C_d>0$ independent of $n$ and $N$.
	
	Next, summing $d_n$ over $n$ yields
	$$
	\sum_{n=1}^N d_n
	=
	\sum_{i\in I_N}\frac{\bar r_i^0}{\bar r_i}\bigl(2\bar r_{i0}+\omega_{i0}-\bar x_{i0}\bigr)
	=
	A+\sum_{i\in I_N}\frac{\bar r_i^0}{\bar r_i}\bar r_{i0}.
	$$
	By Lemma \ref{lem:A-positive}, whenever $N_\beta\ge N_1$, $A\ge c_AI_N>0$. Moreover, by Proposition \ref{prop:equilibrium-bounds}, whenever $N_\beta$ is sufficiently large, one has $\bar r_i^0/\bar r_i\ge 1/2$ and $\bar r_{i0}\ge c_r$. Hence, there exist $N_2\ge N_1$ and $c_d>0$ such that for any $N_\beta\ge N_2$, $\sum_{n=1}^N d_n\ge c_d I_N.$
	
	Let $\bar v_\pi:=\sum_{n=1}^N \pi_n^{N,\beta} v_n$ denote the $\pi$-average of $v$.
	Using \eqref{eq:d-v-orthogonality}, we have
	$$
	\bar v_\pi
	=
	\frac{\sum_{n=1}^N d_n \bar v_\pi}{\sum_{n=1}^N d_n}
	=
	\frac{\sum_{n=1}^N d_n(\bar v_\pi-v_n)}{\sum_{n=1}^N d_n}.
	$$
	Using the bounds on $d_n$, we obtain
	$$
	|\bar v_\pi|
	\le
	\frac{\sum_{n=1}^N |d_n||v_n-\bar v_\pi|}{\sum_{n=1}^N d_n}
	\le
	\frac{C_d}{c_d}\sum_{n=1}^N \pi_n^{N,\beta} |v_n-\bar v_\pi|
	\le
	\frac{C_d}{c_d}\left(\sum_{n=1}^N \pi_n^{N,\beta} (v_n-\bar v_\pi)^2\right)^{1/2},
	$$
	where the last step uses Cauchy--Schwarz and $\sum_{n=1}^N \pi_n^{N,\beta}=1$. Therefore,
	$$
	\norm{v}_{N,\beta}^2=\sum_{n=1}^N \pi_n^{N,\beta} v_n^2
	=
	\sum_{n=1}^N \pi_n^{N,\beta} (v_n-\bar v_\pi)^2+\bar v_\pi^2
	\le
	\left(1+\frac{C_d^2}{c_d^2}\right)\sum_{n=1}^N \pi_n^{N,\beta} (v_n-\bar v_\pi)^2.
	$$
	Hence, we proved that there exists $c_v>0$ such that
	\begin{equation}
		\sum_{n=1}^N \pi_n^{N,\beta} (v_n-\bar v_\pi)^2
		\ge
		c_v\norm{v}_{N,\beta}^2
		\qquad\text{whenever } N_\beta\ge N_2.
		\label{eq:v-variance-lower}
	\end{equation}
	By Lemma \ref{lem:S-structure} and Lemma \ref{lem:w-lower-bound},
	$$
	S(u)
	\ge
	\frac{c_w I_N N_\beta}{2}\sum_{m,n=1}^N \pi_m^{N,\beta}\pi_n^{N,\beta}(v_m-v_n)^2.
	$$
	Since $(\pi_n^{N,\beta})_{n=1}^N$ is a probability vector, the variance identity yields
	$$
	\frac12\sum_{m,n=1}^N \pi_m^{N,\beta}\pi_n^{N,\beta}(v_m-v_n)^2
	=
	\sum_{n=1}^N \pi_n^{N,\beta} (v_n-\bar v_\pi)^2.
	$$
	Therefore, by \eqref{eq:v-variance-lower}, whenever $N_\beta\ge N_2$,
	$$
	S(u)\ge c_w I_N N_\beta \sum_{n=1}^N \pi_n^{N,\beta} (v_n-\bar v_\pi)^2\ge c_SI_NN_\beta \norm{v}_{N,\beta}^2,
	$$
	for some constant $c_S>0$, as claimed.
\end{proof}

\subsection{Control of Income Effects}
In this subsection, we show that the diversification assumption \textnormal{(A5)} controls the distortion-driven income effect so that $M(u)$ grows sublinearly in $N_\beta$.

\begin{proposition}\label{prop:M-upper-bound}
	Suppose \emph{(A1)}--\emph{(A5)} hold. Then there exists a function $o:\mathbb{R}\rightarrow \mathbb{R}$ satisfying $\lim_{x\rightarrow \infty} o(x)/x=0$ such that, for every $N\ge N_0$ and every decomposition $q=\alpha\bar p+u$ with $\Psi(u)=0$ and $v_n:=u_n/\bar p_n$, one has
	\begin{equation}
		|M(u)|\le o(N_\beta) I_N \norm{v}_{N,\beta}^2.
		\label{eq:M-upper-bound}
	\end{equation}
	Here, the function $o$ is independent of the choice of equilibria $(\bar{p}^{N,\beta}, \bar{x}^{N,\beta})$, and $N_0$ is the constant specified in Proposition \ref{prop:equilibrium-bounds}. 
\end{proposition}

\begin{proof}
	At any equilibrium $(\bar{p},\bar{x})$, we write the net trade $t_{in}=\bar x_{in}-\omega_{in}$, and define
	$$
	\Delta_{N,\beta}:=
	\frac{1}{I_N^2}\sum_{i,j\in I_N}
	\langle t_i,t_j\rangle_{N,\beta}\,
	\langle \bar\rho_i,\bar\rho_j\rangle_{N,\beta}.
	$$
	The key observation is that $\Delta_{N,\beta}$ converges to zero when $N_\beta$ goes to infinity, uniformly across all equilibrium selections.
	
	To see this, by Proposition \ref{prop:equilibrium-bounds} and Assumption \textnormal{(A4)}, there exists a constant $C_t>0$, independent of $i$, $n$, and $N$, such that
	$|t_{in}|\le C_t$ for all $i\in I_N,\ n=1,\dots,N.$	Therefore, for every $i,j\in I_N$, whenever $N_\beta\ge N_0$, 
	$$
	\bigl|\langle t_i,t_j\rangle_{N,\beta}\bigr|
	\le
	\frac{1}{N_\beta}\sum_{n=1}^N \beta^n |t_{in}||t_{jn}|
	\le
	\frac{C_t^2}{N_\beta}\sum_{n=1}^N \beta^n
	= C_t^2.
	$$
	Therefore, 
	$$
	\Delta_{N,\beta}
	\le
	C_t^2
	\frac{1}{I_N^2}\sum_{i,j\in I_N}\bigl|\langle \bar\rho_i,\bar\rho_j\rangle_{N,\beta}\bigr| \le C_t^2\kappa(N_\beta),
	$$
	where
	$$\kappa(K)=\sup\left\{\frac{1}{I_N^2}\sum_{i,j\in I_N} |\langle \bar{\rho}_i,\bar{\rho}_j \rangle_{N,\beta}|: N_\beta>K,(\bar{p},\bar{x}) \text{~is ~an~equilibrium~of~}\mathcal{E}^{N,\beta}\right\}.$$
	Note that $\Delta_{N,\beta}$ can depend on the equilibrium selection and, here, we bound it by a quantity $C_t^2\kappa(N_\beta)$ which is independent of the equilibrium selection. It remains to show  $\kappa(K)\rightarrow 0$ as $K\rightarrow \infty$. Suppose this does not hold. Then,  we can find a sequence of equilibria $(\bar{p}^{(k)},\bar{x}^{(k)})$ for $\mathcal{E}^{N^{(k)},\beta^{(k)}}$ such that $N_\beta^{(k)}=\sum_{n=1}^{N^{(k)}}(\beta^{(k)})^n$ goes to infinity and the average alignment  $\frac{1}{I_{N^{(k)}}^2}\sum_{i,j\in I_{N^{(k)}}} |\langle \bar{\rho}_i^{(k)},\bar{\rho}_j^{(k)} \rangle_{N^{(k)},\beta^{(k)}}|\ge c$ for some $c>0$, where $\bar\rho_i^{(k)}$ is defined at the corresponding equilibrium. Then, Assumption \textnormal{(A5)} is violated. Contradiction.
	
	Now, we use this to bound the income effect term $M(u).$ Recall $\bar{m}_n=\frac{1}{I_N} \sum_{i\in I_N} \bar{m}_{in}$, where $\bar{m}_{in}$ is the equilibrium marginal expenditure share of agent $i$ defined in \eqref{eqn:barm_in}. We  can write
	$\bar\Lambda(u):= \frac{1}{I_N} \sum_{i\in I_N} \Lambda_i(u)= \sum_{n=1}^N \bar m_n v_n.$ Then, by the definition of $\bar\rho_{in}$ in \eqref{eqn:relativedeviation}, 
	$$
	\Lambda_i(u)-\bar\Lambda(u)
	=
	\sum_{n=1}^N \bar m_n\bar\rho_{in}v_n
	=
	\langle \bar\rho_i,\tilde v\rangle_{N,\beta},
	$$
	$$
	\sum_{m=1}^N (\omega_{im}-\bar x_{im})u_m
	=
	-\sum_{m=1}^N t_{im}\bar p_m v_m
	=
	-N_\beta\langle t_i,\hat v\rangle_{N,\beta},
	$$
	where $\tilde v_n:=\frac{N_\beta \bar m_n}{\beta^n}v_n,$ and $\hat v_n:=\frac{\bar p_n}{\beta^n}v_n.$ By Proposition \ref{prop:equilibrium-bounds}, \eqref{eqn:barm_in} and \eqref{eq:derived-risk-bounds}, there exist constants $0<c_m < C_m<\infty$ and $0<c_p< C_p<\infty$, such that
	$$
	c_m\frac{\beta^n}{N_\beta}\le \bar m_n\le C_m\frac{\beta^n}{N_\beta},
	\qquad
	c_p\beta^n\le \bar p_n\le C_p\beta^n.
	$$
	Therefore, both $\tilde v_n$ and $\hat v_n$ are uniformly comparable to $v_n$. Formally, taking $C=\max \{C_p,C_m\}$, we have 
	\begin{equation}
		\norm{\tilde{v}}_{N,\beta} \le C\norm{v}_{N,\beta}, \qquad \norm{\hat{v}}_{N,\beta} \le C\norm{v}_{N,\beta}.
		\label{eq:vhattilde-bound}
	\end{equation}
	
	Next, write $\eta_i=\frac{\bar r_{i0}}{\bar r_i}$ and $E_i(u)=\sum_{m=1}^N (\omega_{im}-\bar x_{im})u_m
	=
	-N_\beta\langle t_i,\hat v\rangle_{N,\beta}.$
	Then, by \eqref{eqn:M},
	\begin{equation}
		M(u)
		=
		\sum_{i\in I_N}\Lambda_i(u)E_i(u)
		-
		\sum_{i\in I_N}\eta_i\Lambda_i(u)E_i(u)
		-
		\sum_{i\in I_N}(1-\eta_i)\bar r_{i0}\Lambda_i(u)^2.
		\label{eq:M-decomposition}
	\end{equation}
	
	We first study the leading term of $M(u)$ in \eqref{eq:M-decomposition}. By market clearing,
	$\sum_{i\in I_N}E_i(u)
	=
	\sum_{m=1}^N u_m\sum_{i\in I_N}(\omega_{im}-\bar x_{im})
	=
	0.$ Hence,
	$$
	\sum_{i\in I_N}\Lambda_i(u)E_i(u)
	=
	\sum_{i\in I_N}\bigl(\Lambda_i(u)-\bar\Lambda(u)\bigr)E_i(u)
	=
	-N_\beta\sum_{i\in I_N}\langle \bar\rho_i,\tilde v\rangle_{N,\beta}\langle t_i,\hat v\rangle_{N,\beta}.
	$$
	
	Define an $N\times N$ matrix $B=(B_{mn})$ by
	$B_{mn}:=
	\sum_{i\in I_N}\left(\sqrt{\frac{\beta^m}{N_\beta}}\,t_{im}\right)
	\left(\sqrt{\frac{\beta^n}{N_\beta}}\,\bar\rho_{in}\right).$
	Then
	$$
	\sum_{i\in I_N}\langle \bar\rho_i,\tilde v\rangle_{N,\beta}\langle t_i,\hat v\rangle_{N,\beta}
	=
	\sum_{m,n=1}^N B_{mn}
	\left(\sqrt{\frac{\beta^m}{N_\beta}}\hat v_m\right)
	\left(\sqrt{\frac{\beta^n}{N_\beta}}\tilde v_n\right).
	$$
	Applying Cauchy--Schwarz to the double sum gives
	$$
	\left|
	\sum_{i\in I_N}\langle \bar\rho_i,\tilde v\rangle_{N,\beta}\langle t_i,\hat v\rangle_{N,\beta}
	\right|
	\le
	\left(\sum_{m,n=1}^N B_{mn}^2\right)^{1/2}
	\left( \sum_{m,n=1}^N \frac{\beta^m\beta^n}{N_\beta^2}\hat{v}_m^2\tilde{v}_n^2\right)^{1/2}
	=
	\left(\sum_{m,n=1}^N B_{mn}^2\right)^{1/2}
	\|\hat v\|_{N,\beta}\|\tilde v\|_{N,\beta}.
	$$
	Moreover,
	$$
	\sum_{m,n=1}^N B_{mn}^2
	=
	\sum_{i,j\in I_N}\langle t_i,t_j\rangle_{N,\beta}\langle \bar\rho_i,\bar\rho_j\rangle_{N,\beta} = I_N^2\Delta_{N,\beta} \le C_t^2I_N^2\kappa(N_\beta).
	$$
	Using these two inequalities, with \eqref{eq:vhattilde-bound}, we obtain
	\begin{equation}
		\left|
		\sum_{i\in I_N}\Lambda_i(u)E_i(u)
		\right|
		\le C^2N_\beta\left(\sum_{m,n=1}^N B_{mn}^2\right)^{1/2}\norm{v}_{N,\beta}^2
		\le 
		C^2C_t\sqrt{\kappa(N_\beta)}\,N_\beta I_N\norm{v}_{N,\beta}^2.
		\label{eq:M-main-bound}
	\end{equation}
	
	It remains to bound the last two terms in \eqref{eq:M-decomposition}. 
	First, by \eqref{eq:vhattilde-bound},
	$$
	|E_i(u)|
	=
	N_\beta|\langle t_i,\hat v\rangle_{N,\beta}|
	\le
	N_\beta \norm{t_i}_{N,\beta} \norm{\hat v}_{N,\beta}\le
	C_t N_\beta \norm{\hat v}_{N,\beta}\le C_EN_\beta \norm{v}_{N,\beta},
	$$
	for some $C_E>0$. Moreover, since both $\bar m_{in}$ and $\pi_n^{N,\beta}$ are uniformly comparable to $\beta^n/N_\beta$, one can find $C_{\Lambda}>0$ such that $\bar m_{in}\le C_{\Lambda} \pi_n^{N,\beta}$ for all $i\in I_N$ and $n=1,\dots,N$. Since $(\bar m_{in})_{n=1}^N$ is a probability vector, we obtain
	$$
	|\Lambda_i(u)|
	=
	\left|\sum_{n=1}^N \bar m_{in}v_n\right|
	\le
	\left(\sum_{n=1}^N \bar m_{in}v_n^2\right)^{1/2}
	\le
	C_\Lambda^{1/2}\left(\sum_{n=1}^N \pi_n^{N,\beta} v_n^2\right)^{1/2}=C_\Lambda^{1/2} \norm{v}_{N,\beta}.
	$$
	
	Finally, by Proposition \ref{prop:equilibrium-bounds}, $\bar r_{i0}$ is uniformly bounded and there exists a constant $C_\eta>0$ such that
	$0\le \eta_i\le \frac{C_\eta}{N_\beta}$ for all $i\in I_N$.  Hence, we have
	$$
	\left|
	\sum_{i\in I_N}\eta_i\Lambda_i(u)E_i(u)
	\right|
	+
	\left|
	\sum_{i\in I_N}(1-\eta_i)\bar r_{i0}\Lambda_i(u)^2
	\right|
	\le
	(C_\eta C_E C_\Lambda^{1/2}+C_rC_\Lambda)I_N\norm{v}_{N,\beta}^2.
	$$
	In sum, we have, for some $C_0, C_1>0$, 
	$$
	|M(u)|
	\le
	\left(C_0\sqrt{\kappa(N_\beta)}+\frac{C_1}{N_\beta}\right)
	N_\beta I_N \norm{v}^2_{N,\beta}.
	$$
	Since $\kappa(N_\beta)\to 0$ as $N_\beta\to\infty$, the coefficient $C_0\sqrt{\kappa(N_\beta)}+\frac{C_1}{N_\beta}$ converges to zero as $N_\beta\to\infty$. This proves the proposition.
\end{proof}

\subsection{Proof of the Main Theorem}

We now complete the proof of the main theorem by combining the bounds established in the previous subsections. 

By Lemma \ref{lem:A-positive}, there exist constants $c_A>0$ and $N_1\in\mathbb N$ such that
$$
A\ge c_A I_N>0
\qquad\text{whenever } N_\beta\ge N_1.
$$
Hence, by Proposition \ref{prop:quadratic-decomposition}, the decomposition $q=\alpha \bar p+u$ for $u\in\ker\Psi$ is uniquely defined for every $q\in \mathbb{R}^N$. Write $v_n=u_n/\bar{p}_n$ for each $n=1,\dots,N$.

Next, Proposition \ref{prop:S-lower-bound} yields a constant $c_S>0$ and an integer $N_2\ge N_1$ such that, for every $u\in\ker\Psi$,
$$
S(u)\ge c_S N_\beta I_N \norm{v}_{N,\beta}^2
\qquad\text{whenever } N_\beta\ge N_2.
$$
Proposition \ref{prop:M-upper-bound} implies that there exists an integer $N_3\ge N_2$ such that, for every $u\in\ker\Psi$,
$$
|M(u)|\le \frac{c_S}{2} N_\beta I_N  \norm{v}_{N,\beta}^2
\qquad\text{whenever } N_\beta\ge N_3.
$$
Therefore, whenever $N_\beta\ge N_3$, for any price perturbation $q=\alpha\bar p+u\neq 0$ with $u\in\ker\Psi$,
\begin{equation}
	q^T Dz(\bar{p}) q
	\le
	-A\alpha^2+\alpha R(u)-\frac{c_S}{2}N_\beta I_N\norm{v}_{N,\beta}^2.
	\label{eq:main-proof-reduced-form}
\end{equation}
Finally, by Lemma \ref{lem:R-order-one}, there exists a constant $C_R>0$ such that, for every $u\in\mathbb R^N$, $|R(u)|\le C_R I_N\norm{v}_{N,\beta}$ whenever $N_\beta\ge N_3\ge N_0$. Then \eqref{eq:main-proof-reduced-form} implies
$$
q^T Dz(\bar{p}) q
\le
-c_A I_N\alpha^2+C_R I_N\norm{v}_{N,\beta} |\alpha|-\frac{c_S}{2}N_\beta I_N \norm{v}_{N,\beta}^2.
$$
If $u=0$, then $q=\alpha\bar p$ with $\alpha\neq 0$, so
$$
q^T Dz(\bar{p}) q\le -c_A I_N \alpha^2<0.
$$
If $u\neq 0$, then $v\neq 0$ implies $\norm{v}_{N,\beta}>0$. The right-hand side is a concave quadratic function of $\alpha$, and its maximum is attained at
$|\alpha|=\frac{C_R \norm{v}_{N,\beta}}{2c_A}.$
Hence,
$$
q^T Dz(\bar{p}) q
\le
\left(
\frac{C_R^2}{4c_A}-\frac{c_S}{2}N_\beta
\right)I_N\norm{v}_{N,\beta}^2.
$$
Therefore, $q^T Dz(\bar{p}) q<0$ for every $q\neq 0$ whenever 
$N_\beta>\frac{C_R^2}{2c_Ac_S}.$

Finally, let
$N^*:=\max\left\{N_3,\frac{C_R^2}{2c_Ac_S}\right\}.$
Then, whenever $N_\beta>N^*$, every interior equilibrium price $\bar p$ of $\mathcal E^{N,\beta}$ satisfies $q^T Dz(\bar p)q<0$ for all $q\in\mathbb R^N\setminus\{0\}$. Hence, $\bar p$ is locally t\^atonnement stable. Since every equilibrium is interior by Lemma \ref{lem:equilibrium-interior}, this proves the theorem.

\section{Discussion}\label{section:discussion}

\subsection{Discounting and Commodity Space} \label{subsec:choiceofutilityform}

In our analysis, we assume that agents' utilities take the discounted-sum form
$$
U_i(x_i)=\sum_{n=0}^N \beta^n u_{in}(x_{in}).
$$
This means that all agents disregard large-indexed commodities asymptotically, and do so at a common rate. Moreover, there is a single commodity at each date.

The feature of discarding large-indexed commodities asymptotically is motivated by the classical study of $\ell^\infty$-economies with countably many commodities. In that literature, Mackey continuity, introduced by \citet*{bewley1972existence}, is a standard regularity condition in equilibrium existence theory and captures the idea that tail commodities become asymptotically irrelevant. \citet*{brown1981myopic} interpret this property as a form of myopia, \citet*{stroyan1983myopic} develops the same idea in sequential economies, and \citet*{araujo1985lack} shows that impatience may be essential in such economies. See also \citet*{bastianello2017topological} for a recent reference. Thus, although our analysis would still apply if $\beta=1$, the resulting sequence of economies would no longer approximate the classical countable-commodity framework with tail irrelevance. More generally, the discount factor $\beta^n$ can be replaced by any positive summable sequence $b_n$  with sufficiently large total mass.

Another simplifying feature of the model is that all agents share the same
patience parameter. The role of this assumption is mainly expositional. With
heterogeneous patience parameters, agents whose patience is bounded away from
the maximal patience level discount sufficiently distant commodities more
heavily and therefore allocate relatively more of their consumption to earlier
periods. For such agents, adding further commodities far in the tail has a
vanishing effect on demand. Thus, the demand behavior on tail commodities is governed by the agents with maximal patience. We therefore impose a common patience
parameter in order to study directly the class of agents that determines the
asymptotic effect of enlarging the commodity space.

Finally, it is straightforward to extend the framework from one commodity at each date to additively separable per-period utility with multiple commodities. For instance, with two commodities in each period, one can relabel the goods as a single sequence of commodities and work with a discount factor $\tilde\beta$ satisfying $\tilde\beta^2=\beta$. Thus, the one-good-per-date formulation is adopted only for notational simplicity.

\subsection{A No-Concentration Interpretation of Assumption (A5)}
\label{subsection:assumptiondiscussion}

We now explain how Assumption \textnormal{(A5)} can be viewed as a no-concentration condition. The interpretation is especially transparent in a logarithmic environment, where $u_{in}(x)=\tau_{in}\log x,$ and the taste parameters are normalized so that
$\sum_{n=1}^N \beta^n\tau_{in}=1$ for each $i\in I_N$. In this case, for any allocation $x=(x_i)_{i\in I_N}$,
$$
m_{in}(x_i)
=
\frac{\beta^n\tau_{in}}{\sum_{m=1}^N\beta^m\tau_{im}}
=
\beta^n\tau_{in}, \qquad \rho_{in}(x)=\frac{m_{in}-\bar m_n}{\bar m_n}=\frac{\tau_{in}}{\bar\tau_n}-1,
$$
where $
\bar\tau_n:=\frac{1}{I_N}\sum_{j\in I_N}\tau_{jn}.$ Thus, in the logarithmic benchmark, \textnormal{(A5)} becomes a condition on taste-deviation profiles $\rho_i=
(
\frac{\tau_{in}}{\bar\tau_n}-1
)_{n=1}^N,$ satisfying
\begin{equation}
\frac{1}{I_N^2}\sum_{i,j\in I_N}
\left|
\frac{1}{N_\beta}\sum_{n=1}^N \beta^n
\left(
\frac{\tau_{in}}{\bar\tau_n}-1
\right)
\left(
\frac{\tau_{jn}}{\bar\tau_n}-1
\right)
\right|
\rightarrow 0 \qquad \text{~as~} N_\beta\rightarrow \infty.
\label{eqn: A5-log}
\tag{A5--log}
\end{equation}

One convenient sufficient condition for \eqref{eqn: A5-log} is spectral. Define matrices 
$D=\frac{1}{N_\beta}\operatorname{diag}(\beta,\beta^2,\dots,\beta^N)$ and 
$\Sigma=\frac{1}{I_N}\sum_{i\in I_N}\rho_i(\rho_i)^T.$ Then, by Cauchy--Schwarz,
$$\left(\frac{1}{I_N^2}\sum_{i,j\in I_N}
\left|\langle \rho_i,\rho_j\rangle_{N,\beta}\right|\right)^2\le 
\frac{1}{I_N^2}\sum_{i,j\in I_N}
\langle \rho_i,\rho_j\rangle_{N,\beta}^2
=
\operatorname{tr}\bigl((D\Sigma)^2\bigr)
=
\operatorname{tr}\bigl((D^{1/2}\Sigma D^{1/2})^2\bigr)
=
\sum_{k=1}^N\lambda_k^2,
$$
where $\lambda_1,\dots,\lambda_N$ are the eigenvalues of
$D^{1/2}\Sigma D^{1/2}$. Hence, \eqref{eqn: A5-log} is implied by the spectral condition
$$\lim_{N_\beta\rightarrow \infty}\sum_{k=1}^N\lambda_k^2=0.$$ 

This spectral sufficient condition admits a no-concentration interpretation.  If $v_k$ is a
unit eigenvector associated with $\lambda_k$, then
$$
\lambda_k
=
v_k^TD^{1/2}\Sigma D^{1/2}v_k
=
\frac{1}{I_N}\sum_{i\in I_N}
\bigl(v_k^TD^{1/2}\rho_i\bigr)^2 .
$$
Thus, $\lambda_k$ measures the average squared projection of the weighted
relative taste-deviation profiles $D^{1/2}\rho_i$ onto the direction
$v_k$. Large eigenvalues indicate that these profiles are concentrated along a
small number of common directions. Thus, when agents' taste deviations from the population average do not concentrate around only a few directions,  \textnormal{(A5)} is satisfied.

\subsection{Two Forms of Preference Diversification}
\label{subsection:examples}

This subsection presents two forms of preference diversification under which Assumption \textnormal{(A5)} holds. As a benchmark, consider the degenerate case in which all agents have the same equilibrium marginal expenditure-share profile, so that $\rho_i=0$ for every $i$. This includes Gorman-form aggregation. If each agent's demand for future commodities satisfies $\xi_i(p,w_i)=a_i(p)+b(p)w_i,$
with the wealth slope $b(p)$ common across agents, then the marginal expenditure-share profile is identical across agents, and \textnormal{(A5)} holds trivially. In our setting, common isoelastic utility,
$u_{in}(x)=\tau_n\frac{x^{1-1/\sigma}}{1-1/\sigma}$,
is a homothetic special case. In this benchmark, the common future-wealth income-effect component cancels at equilibrium by market clearing, so the large-commodity argument is not needed for establishing equilibrium stability.

We now provide two non-degenerate scenarios in which \textnormal{(A5)} holds when the number of commodities is large. The first scenario is one in which agents' preferences exhibit sparse or localized deviations as the number of commodities grows. This scenario approximates the degenerate benchmark, in the sense that agents' marginal expenditure-share profiles are similar because their preferences differ on only a limited set of commodities. To illustrate, suppose agents' utilities are given by the logarithmic benchmark $u_{in}(x)=\tau_{in}\log x$ studied in the previous subsection. In this case, by the Cauchy--Schwarz inequality, a convenient sufficient condition for \eqref{eqn: A5-log} is
$$
\sup_{i\in I_N}\frac{1}{N_\beta}\sum_{n=1}^N\beta^n
\left(\frac{\tau_{in}}{\bar\tau_n}-1\right)^2\to 0 \qquad \text{~as~} N_\beta\rightarrow \infty.
$$
This covers environments in which each agent's taste differs from the population average only on a $\beta$-weighted negligible set of commodities. Such a pattern may arise under sparse taste deviations or in Hotelling-type environments, where agents have relatively high taste weights on nearly disjoint blocks of commodities.

The second scenario is one in which agents' preferences exhibit dispersed heterogeneity as the number of commodities grows. This scenario is natural when the set of commodity types expands and agents' preference types become correspondingly more diverse. To illustrate, again suppose agents' utilities are given by the logarithmic benchmark $u_{in}(x)=\tau_{in}\log x$ studied in the previous subsection. Moreover, suppose agents' tastes are parametrized by $\tau_{in}=\frac{1}{N_\beta}(1+\delta\varepsilon_{in})$, where the deviations $\varepsilon_{in}$ are balanced across dates,
$\frac{1}{N_\beta}\sum_{n=1}^N \beta^n\varepsilon_{in}=0$
for each $i$, and have mean zero across agents at each date,
$\frac{1}{I_N}\sum_{i\in I_N} \varepsilon_{in}=0$
for every $n$. Then $\rho_{in}=\delta\varepsilon_{in}$, and \textnormal{(A5)} becomes
$$
\frac{1}{I_N^2}\sum_{i,j\in I_N}
\left|
\frac{1}{N_\beta}\sum_{n=1}^N\beta^n\varepsilon_{in}\varepsilon_{jn}
\right|\to 0 \qquad \text{~as~} N_\beta\rightarrow \infty.
$$
Thus, \textnormal{(A5)} holds when these balanced deviation profiles are sufficiently diffuse across agents.

\subsection{Necessity of Preference Diversification} 
\label{subsection:necessity}

Some form of diversification condition such as \textnormal{(A5)} is needed for our asymptotic comparison between substitution and income effects: without such a condition, the income effect can have the same order of magnitude as the substitution effect. To see this, consider a simple two-type logarithmic economy. Suppose endowments are identical across agents and dates, and half of the agents are of type $A$ while half are of type $B$. Let
$$
u_{An}(x)=(1+\delta\varepsilon_n)\log x,
\qquad
u_{Bn}(x)=(1-\delta\varepsilon_n)\log x,
$$
where $\delta\in(0,1)$, $\varepsilon_n=1$ on odd dates, and
$\varepsilon_n=-1$ on even dates. Since utility is logarithmic, the marginal
expenditure-share profiles are independent of equilibrium allocations. One can show that the two types have opposite odd-even distortions:
$$
\rho_{A n}=\delta\varepsilon_n+o(1),
\qquad
\rho_{B n}=-\delta\varepsilon_n+o(1).
$$
Therefore, the average weighted alignment in Assumption \textnormal{(A5)} remains
bounded away from zero, so the assumption is violated. In this case, the vectors $\rho_i$ are concentrated around only two directions.

Now consider a price distortion whose normalized component is this same
odd-even pattern, with $q_n=\varepsilon_n \bar p_n$. Indeed, since $\Lambda_i(q)=\sum_n m_{in}v_n$, we have
$$
\Lambda_A(q)
=
\frac{E^{N,\beta}+\delta N_\beta}{N_\beta+\delta E^{N,\beta}}
=
\delta+O(N_\beta^{-1}),
\qquad
\Lambda_B(q)
=
\frac{E^{N,\beta}-\delta N_\beta}{N_\beta-\delta E^{N,\beta}}
=
-\delta+O(N_\beta^{-1}),
$$
where $E^{N,\beta}=\sum_{n=1}^N\beta^n\varepsilon_n=O(1)$. Thus
$\Lambda_A(q)$ and $\Lambda_B(q)$ are both of order one and have opposite
signs. Moreover, type $A$ buys the commodities for which
$\varepsilon_n=1$ and sells those for which $\varepsilon_n=-1$, while type
$B$ does the reverse. Hence the net-trade term
$\sum_{n=1}^N(\omega_{in}-\bar x_{in})q_n$ is of order $N_\beta$ for each type. The product of this term with
$\Lambda_i(q)$ therefore remains of order $N_\beta$, and the contributions
of the two types have the same sign rather than canceling. Consequently, the
aggregate income term $M(q)$ is of order $I_NN_\beta$, the same order as the
substitution term.

\subsection{Comparison with Classical Stability Conditions}\label{subsec:comparisons}

We now illustrate, through an isoelastic example, that our stability conditions differ from three types of classical conditions. Consider an economy in which all agents have isoelastic utilities with common elasticity $\sigma>0$, $\sigma\neq1$:
$$
u_{i0}(x)=\frac{x^{1-1/\sigma}}{1-1/\sigma},\qquad
u_{in}(x)=(1+\delta\varepsilon_{in})^{1/\sigma}
\frac{x^{1-1/\sigma}}{1-1/\sigma},\quad n=1,\dots,N,
$$
where $\delta\in(0,1)$, the shock parameters $\varepsilon_{in}$ are uniformly bounded, and $1+\delta\varepsilon_{in}$ is uniformly bounded away from zero. Let the endowment be $\omega_{in}=\bar\omega+\eta_i s_n$ for $n=0,\dots,N,$
where $\bar\omega>0$, $0<s_n<\bar s<\bar\omega$, and $\eta_i\in\{-1,1\}$. Assume
$$
\frac{1}{N_\beta}\sum_{n=1}^N\beta^n\varepsilon_{in}=0,\qquad
\frac{1}{I_N}\sum_{i\in I_N}\varepsilon_{in}=0,\qquad
\frac{1}{I_N}\sum_{i\in I_N}\eta_i=0,\qquad
\frac{1}{I_N}\sum_{i\in I_N}\eta_i\varepsilon_{in}=0
$$
for all $i,n$. Thus, Assumptions \textnormal{(A1)}--\textnormal{(A4)} hold.

Suppose, in addition, that
$$
\lim_{N_\beta\to\infty} \frac{1}{I_N^2}\sum_{i,j\in I_N}
\left|
\frac{1}{N_\beta}\sum_{n=1}^N
\beta^n\varepsilon_{in}\varepsilon_{jn}
\right|= 0.
$$
This is the dispersed-heterogeneity condition from the previous subsection. \ref{appendix:isoelasticexample} shows that, for all sufficiently large $N_\beta$, every equilibrium $(\bar p,\bar x)$ satisfies
$$
\bar p_n=\beta^n,\qquad
\bar x_{i0}=\bar\omega+\eta_i s^{N,\beta},\qquad
\bar x_{in}=(\bar\omega+\eta_i s^{N,\beta})(1+\delta\varepsilon_{in}),
$$
where $s^{N,\beta}=\frac{s_0+\sum_{n=1}^N\beta^n s_n}{1+N_\beta}.$ Therefore,
$$
m_{in}(\bar x_i)
=
\frac{\beta^n}{N_\beta}(1+\delta\varepsilon_{in}),
\qquad
\bar m_n(\bar x)=\frac{\beta^n}{N_\beta},
\qquad
\bar\rho_{in}=\delta\varepsilon_{in}.
$$
Hence, Assumption \textnormal{(A5)} holds.

Now choose $\sigma<1/4$, and choose $(s_n)$ so that $s^{N,\beta}$ does not converge as $N_\beta\to\infty$. Then this example is outside the scope of the following three classical routes to stability.

First, it violates gross substitutability. For additively separable isoelastic utility, gross substitutability is satisfied only when the elasticity is at least one. Hence the example with $\sigma<1/4$ violates the gross-substitutes condition of \citet*{arrow1959stability}.

Second, agents do not behave as if the shadow value of money were asymptotically fixed, a feature implied by the permanent-income structure studied by \citet*{bewley1980permanent1}. In the example, the shadow value of money is
$\lambda_i=(\bar\omega+\eta_i s^{N,\beta})^{-1/\sigma}.$
If $s^{N,\beta}$ does not converge, then the shadow value of money for each agent need not converge as the number of commodities grows.

Third, it can violate the curvature restrictions used to obtain the law of demand in the tradition of \citet*{milleron1974unicite}, \citet*{polterovich1978criteria}, and \citet*{quah2000monotonicity}. For isoelastic utility,
$$
-\frac{x u_{in}''(x)}{u_{in}'(x)}=\frac{1}{\sigma}.
$$
Thus, when $\sigma<1/4$, relative curvature exceeds the canonical Mitjushin--Polterovich bound $4$.

\section{Concluding Remarks}\label{sec:conclude}

This paper studies how a large effective number of commodities can generate local t\^atonnement stability and, through the index theorem, equilibrium uniqueness in a dated-commodity exchange economy. The key mechanism is that the substitution effect accumulates naturally at the rate of the effective number of commodities, while the income effect is controlled to grow at a slower rate under a preference diversification assumption, which seems plausible under commodity differentiation.

One key technical difficulty tackled in this paper comes from the discounting in the utility formulation. This feature is necessary if the finite-horizon truncations are to approximate a countable-commodity economy with Mackey-continuous preferences, but it also implies that equilibrium prices decline along the horizon and that the corresponding substitution terms become small as the number of commodities grows. For this reason, it may be cleaner to develop the same idea on a compact space of commodity characteristics, where equilibrium prices can be kept uniformly bounded. Nevertheless, that approach would require working with differential equations with a continuum of variables, and would therefore lead to a rather different set of issues.

\appendix
\renewcommand{\theequation}{A.\arabic{equation}}
\renewcommand{\thelemma}{A.\arabic{lemma}} 
\setcounter{equation}{0}
\setcounter{lemma}{0}

\section{Omitted Proofs}

\subsection{Marginal Expenditure Shares} \label{appendix:marginal-spending-profile}

We first state a regularity result. It allows us to treat each agent's demand and Lagrange multiplier as continuously differentiable functions of prices and wealth near the optimizer.

\begin{lemma}\label{lem:local-differentiability}
	Suppose Assumption \textnormal{(A1)} holds and $\bar{x}_i$ is the optimal consumption of agent $i$ under price $\bar{p}$. Define the corresponding wealth as $\bar w_i:=\omega_{i0}+\sum_{n=1}^N \bar p_n\omega_{in}.$ Then, there exist an open neighborhood $P_i\subset\mathbb R_{++}^N$ of $\bar p$, an open neighborhood $W_i\subset\mathbb R_{++}$ of $\bar w_i$, and continuously differentiable functions
	$\xi_{i}:P_i\times W_i\to\mathbb R_{++}^{N+1}$ and $\lambda_{i}:P_i\times W_i\to\mathbb R_{++}$ such that, for every $(p,w_i)\in P_i\times W_i$, the pair $(\xi_{i}(p,w_i),\lambda_i(p,w_i))$ satisfies
	\begin{equation}\label{eq:FOC-good0}
		u_{i0}'(\xi_{i0}(p,w_i))=\lambda_{i}(p,w_i),
	\end{equation}
	\begin{equation}\label{eq:FOC-goodn}
		\beta^n u_{in}'(\xi_{in}(p,w_i))=\lambda_{i}(p,w_i)p_n,
		\qquad n=1,\dots,N,
	\end{equation}
	and
	\begin{equation}\label{eq:budget-local}
		\xi_{i0}(p,w_i)+\sum_{n=1}^N p_n\xi_{in}(p,w_i)=w_i.
	\end{equation}
	In particular, the function $\xi_i(p,w_i)=(\xi_{in}(p,w_i))_{n=0}^N$ coincides with agent $i$'s Marshallian demand $\xi_i(p)$, provided $w_i=\omega_{i0}+\sum_{n=1}^N p_n\omega_{in}$.
\end{lemma}
\begin{proof}
	Define the map $F_i:\mathbb{R}_{++}^{N+1}\times \mathbb{R}_{++}\times \mathbb{R}_{++}^N\times \mathbb{R}_{++}\to \mathbb{R}^{N+2}$ by
	$$
	F_i(x_i,\lambda_i,p,w_i)
	=
	\begin{pmatrix}
		u_{i0}'(x_{i0})-\lambda_i\\
		\beta u_{i1}'(x_{i1})-\lambda_i p_1\\
		\vdots\\
		\beta^N u_{iN}'(x_{iN})-\lambda_i p_N\\
		x_{i0}+\sum_{n=1}^N p_n x_{in}-w_i
	\end{pmatrix}.
	$$
	By Assumption \textnormal{(A1)}, $u_{in}$ is twice continuously differentiable on $\mathbb R_{++}$, so $F_i$ is a $C^1$ map.
	
	Let $\bar x_i=(\bar x_{in})_{n=0}^N$ be agent $i$'s optimal consumption at $\bar p$, let $\bar w_i=\omega_{i0}+\sum_{n=1}^N \bar p_n\omega_{in}$, and let $\bar\lambda_i>0$ be the corresponding Lagrange multiplier. By the first-order conditions for agent $i$'s utility maximization problem and the budget identity, we have
	$$
	F_i(\bar x_i,\bar\lambda_i,\bar p,\bar w_i)=0.
	$$
	
	We next show that the Jacobian of $F_i$ with respect to $(x_i,\lambda_i)$ is invertible at $(\bar x_i,\bar\lambda_i,\bar p,\bar w_i)$. This Jacobian is
	$$
	D_{(x_i,\lambda_i)}F_i(\bar x_i,\bar\lambda_i,\bar p,\bar w_i)
	=
	\begin{pmatrix}
		u_{i0}''(\bar x_{i0}) & 0 & \cdots & 0 & -1\\
		0 & \beta u_{i1}''(\bar x_{i1}) & \cdots & 0 & -\bar p_1\\
		\vdots & \vdots & \ddots & \vdots & \vdots\\
		0 & 0 & \cdots & \beta^N u_{iN}''(\bar x_{iN}) & -\bar p_N\\
		1 & \bar p_1 & \cdots & \bar p_N & 0
	\end{pmatrix}.
	$$
	Suppose $(\zeta_i,\mu)\in \mathbb R^{N+1}\times \mathbb R$ belongs to its kernel. Then
	$$
	u_{i0}''(\bar x_{i0})\zeta_{i0}-\mu=0,\qquad
	\beta^n u_{in}''(\bar x_{in})\zeta_{in}-\bar p_n\mu=0 \ \text{for } n=1,\dots,N,
	$$
	and
	$$
	\zeta_{i0}+\sum_{n=1}^N \bar p_n\zeta_{in}=0.
	$$
	Hence
	$$
	\zeta_{i0}=\frac{\mu}{u_{i0}''(\bar x_{i0})},
	\qquad
	\zeta_{in}=\frac{\bar p_n\mu}{\beta^n u_{in}''(\bar x_{in})}\quad \text{for } n=1,\dots,N.
	$$
	Substituting these expressions into the budget equation gives
	$$
	\left(
	\frac{1}{u_{i0}''(\bar x_{i0})}
	+
	\sum_{n=1}^N \frac{\bar p_n^2}{\beta^n u_{in}''(\bar x_{in})}
	\right)\mu=0.
	$$
	Since $u_{in}''<0$ on $\mathbb R_{++}$ for every $n$ and $\bar p_n>0$, the coefficient in parentheses is strictly negative. It follows that $\mu=0$, and then $\zeta_{in}=0$ for every $n=0,1,\dots,N$. Thus the kernel is trivial, so $D_{(x_i,\lambda_i)}F_i(\bar x_i,\bar\lambda_i,\bar p,\bar w_i)$ is invertible.
	
	The implicit function theorem therefore yields open neighborhoods $P_i\subset \mathbb R_{++}^N$ of $\bar p$, $W_i\subset \mathbb R_{++}$ of $\bar w_i$, and continuously differentiable functions
	$\xi_i: P_i\times  W_i\to \mathbb R_{++}^{N+1},$ and 
	$\lambda_i: P_i\times  W_i\to \mathbb R_{++},$
	such that $(\xi_i(\bar p,\bar w_i),\lambda_i(\bar p,\bar w_i))=(\bar x_i,\bar\lambda_i)$ and
	$$
	F_i(\xi_i(p,w_i),\lambda_i(p,w_i),p,w_i)=0
	$$
	for every $(p,w_i)\in  P_i\times W_i$. Writing $\xi_i(p,w_i)=(\xi_{in}(p,w_i))_{n=0}^N$, this means that, for every $(p,w_i)\in  P_i\times  W_i$, the pair $(\xi_i(p,w_i),\lambda_i(p,w_i))$ satisfies \eqref{eq:FOC-good0}--\eqref{eq:budget-local}.
	
	Finally, if $w_i=\omega_{i0}+\sum_{n=1}^N p_n\omega_{in}$, then $(\xi_i(p,w_i),\lambda_i(p,w_i))$ satisfies the first-order conditions and budget equation for agent $i$'s utility maximization problem at prices $p$ and wealth $w_i$. By strict concavity of preferences under Assumption \textnormal{(A1)}, this maximization problem has a unique solution. Hence $\xi_i(p,w_i)$ coincides with agent $i$'s Marshallian demand. This proves the lemma.
\end{proof}

Now, we explain why we interpret $m_{in}(x_i)$ as a marginal expenditure share.

\begin{lemma}\label{lem:marginal-spending-profile}
	Suppose \textnormal{(A1)} holds. Fix $N\ge 1$, future prices $p=(p_1,\dots,p_N)\gg 0$, and future wealth $w_i>0$. Consider agent $i$'s future utility-maximization problem
	$$
	\max_{(z_1,\dots,z_N)\in \mathbb R_+^N}
	\sum_{n=1}^N \beta^n u_{in}(z_n)
	\qquad\text{subject to}\qquad
	\sum_{n=1}^N p_n z_n\le w_i.
	$$
	Suppose $x_i=(x_{i1},\dots,x_{iN})$ is an interior solution. Then, for every $n=1,\dots,N$,
	\begin{equation*}
		m_{in}(x_i)=p_n\frac{\partial x_{in}}{\partial w_i}.
	\end{equation*}
\end{lemma}
This lemma shows that $m_{in}(x_i)$ is the marginal expenditure share on date-$n$ commodity out of an additional unit of future wealth, provided $x_i$ is an optimal consumption bundle. Moreover, if $(\bar p,\bar x)$ is an equilibrium of $\mathcal E^{N,\beta}$, then $\bar m_{in}=m_{in}(\bar x_i)$ is its equilibrium counterpart.
\begin{proof}
	Let $\lambda_i$ be the multiplier on the budget constraint. Since $x_i$ is an interior optimum, the first-order conditions are
	$$
	\beta^n u_{in}'(x_{in})=\lambda_i p_n,\qquad n=1,\dots,N.
	$$
	Differentiating these first-order conditions with respect to $w_i$ yields 
	$$
	\beta^n u_{in}''(x_{in})\frac{\partial x_{in}}{\partial w_i}
	=
	p_n\frac{\partial \lambda_i}{\partial w_i},\qquad n=1,\dots,N.
	$$
	Here, the relevant differentiabilities are verified in Lemma \ref{lem:local-differentiability}. Using the definition of risk tolerance, we obtain
	$$
	p_n\frac{\partial x_{in}}{\partial w_i}
	=
	-p_n r_{in}(x_{in})\frac{\partial \lambda_i/\partial w_i}{\lambda_i},\qquad n=1,\dots,N.
	$$
	
	Since the budget binds at an interior optimum, differentiating $\sum_{n=1}^N p_n x_{in}=w_i$ with respect to $w_i$ gives
	$$
	\sum_{n=1}^N p_n\frac{\partial x_{in}}{\partial w_i}=1.
	$$
	Substituting the previous expression into this identity, we obtain
	$$
	-\frac{\partial \lambda_i/\partial w_i}{\lambda_i}
	=
	\frac{1}{\sum_{m=1}^N p_m r_{im}(x_{im})}
	=
	\frac{1}{r_i^0(x_i)}.
	$$
	Hence, using the first-order conditions to substitute $p_m=\beta^m u_{im}'(x_{im})/\lambda_i$ for each $m$, we have
	$$
	p_n\frac{\partial x_{in}}{\partial w_i}
	=
	\frac{p_n r_{in}(x_{in})}{r_i^0(x_i)}
	= 	\frac{\beta^n u_{in}'(x_{in})\,r_{in}(x_{in})}
	{\sum_{m=1}^N \beta^m u_{im}'(x_{im})\,r_{im}(x_{im})}=
	{m}_{in}(x_i),
	$$
	which proves the lemma.
\end{proof}

\subsection{Proof of Lemma \ref{lem:equilibrium-interior}} \label{proof:interiorequilibrium}

Let $(\bar p,(\bar x_i)_{i\in I_N})$ be an equilibrium of $\mathcal E^{N,\beta}$, and  suppose instead that $\bar p_n=0$ for some $n\in\{1,\dots,N\}$. Fix any agent $i\in I_N$. Since $\bar x_i$ is optimal for agent $i$ at prices $\bar p$, it must maximize $U_i(x_i)=\sum_{m=0}^N \beta^m u_{in}(x_{im})$ subject to the budget constraint $x_{i0}+\sum_{m=1}^N \bar p_m x_{im}\le \omega_{i0}+\sum_{m=1}^N \bar p_m\omega_{im}$. But because $\bar p_n=0$, for every $t>0$ the bundle $\bar x_i+t e_n$ remains affordable, where $e_n$ is the unit vector for commodity $n$. Moreover, by Assumption \textnormal{(A1)}, $u_{in}$ is strictly increasing on $\mathbb R_+$. Hence $U_i(\bar x_i+t e_n)>U_i(\bar x_i)$ for $t>0$. This contradicts the optimality of $\bar x_i$. 

\subsection{Proof of Lemma \ref{lem:substitution-income-decomposition}}\label{proof:sluskydecom}

When no confusion can arise, we use $\xi_i(p,w_i)$ for the local demand map given by Lemma \ref{lem:local-differentiability}, and we continue to write $\xi_i(p)$ for Marshallian demand in the exchange economy $\mathcal E^{N,\beta}$. We also distinguish the total derivative $d\xi_{in}/dp_m$ from the partial derivative $\partial \xi_{in}/\partial p_m$, where wealth is held fixed.

By Lemma \ref{lem:equilibrium-interior}, we know $\bar{p}\gg 0$. Fix $i\in I_N$. First, hold wealth fixed at $\bar w_i$ and vary only price $p_m$. Consider the path $p(t)=\bar p+t e_m$, where $e_m$ is the $m$-th standard basis vector of $\mathbb R^N$. Write
$$x_{in}(t)=\xi_{in}(p(t),\bar w_i),\qquad \lambda_i(t)=\lambda_i(p(t),\bar w_i).$$
By Lemma \ref{lem:local-differentiability}, these functions are continuously differentiable for $t$ near $0$. They satisfy
$$u_{i0}'(x_{i0}(t))=\lambda_i(t),\qquad
\beta^n u_{in}'(x_{in}(t))=\lambda_i(t)p_n(t)\quad \text{for } n=1,\dots,N,$$
together with the budget equation
$$x_{i0}(t)+\sum_{n=1}^N p_n(t)x_{in}(t)=\bar w_i.$$

Differentiating at $t=0$, we obtain
$$
\frac{d x_{i0}}{dt}(0)=\frac{1}{u_{i0}''(\bar x_{i0})}\frac{d\lambda_i}{dt}(0)
=-\frac{\bar r_{i0}}{\bar\lambda_i}\frac{d\lambda_i}{dt}(0).
$$
For each $n=1,\dots,N$, differentiating the $n$-th first-order condition gives
$$
\beta^n u_{in}''(\bar x_{in})\frac{d x_{in}}{dt}(0)
=
\bar p_n\frac{d\lambda_i}{dt}(0)+\bar\lambda_i \mathbf 1_{\{n=m\}}.
$$
Using $-\frac{\bar p_n}{\beta^n u_{in}''(\bar x_{in})}=\frac{\bar r_{in}}{\bar\lambda_i}$ and
$-\frac{\bar\lambda_i}{\beta^n u_{in}''(\bar x_{in})}=\frac{\bar r_{in}}{\bar p_n}$, this becomes
$$
\frac{d x_{in}}{dt}(0)
=
-\frac{\bar r_{in}}{\bar\lambda_i}\frac{d\lambda_i}{dt}(0)
-\frac{\bar r_{in}}{\bar p_n}\mathbf 1_{\{n=m\}}.
$$

Now differentiate the budget equation at $t=0$. Since $\frac{d p_n}{dt}(0)=\mathbf 1_{\{n=m\}}$, we get
$$
\frac{d x_{i0}}{dt}(0)+\sum_{n=1}^N \bar p_n\frac{d x_{in}}{dt}(0)+\bar x_{im}=0.
$$
Substituting the expressions above yields
$$
-\frac{\bar r_{i0}+\sum_{n=1}^N \bar p_n \bar r_{in}}{\bar\lambda_i}\frac{d\lambda_i}{dt}(0)-\bar r_{im}+\bar x_{im}=0.
$$
Using $\bar r_i=\bar r_{i0}+\sum_{n=1}^N \bar p_n \bar r_{in}$, we obtain
$$
\frac{\bar r_i}{\bar\lambda_i}\frac{d\lambda_i}{dt}(0)=\bar x_{im}-\bar r_{im},
\qquad\text{so}\qquad
\frac{d\lambda_i}{dt}(0)=\frac{\bar\lambda_i}{\bar r_i}(\bar x_{im}-\bar r_{im}).
$$
Substituting this back gives
\begin{equation}\label{eq:proof-fixed-wealth}
	\frac{\partial \xi_{in}}{\partial p_m}(\bar p,\bar w_i)
	=
	\frac{d x_{in}}{dt}(0)
	=
	\frac{\bar r_{in}\bar r_{im}}{\bar r_i}
	-\frac{\bar r_{in}}{\bar r_i}\bar x_{im}
	-\frac{\bar r_{in}}{\bar p_n}\mathbf 1_{\{n=m\}}.
\end{equation}

Next, hold prices fixed at $\bar p$ and vary wealth. Consider the path $w_i(t)=\bar w_i+t$, and write
$$
y_{in}(t)=\xi_{in}(\bar p,w_i(t)),\qquad \mu_i(t)=\lambda_i(\bar p,w_i(t)).
$$
Again by Lemma \ref{lem:local-differentiability}, these functions are continuously differentiable near $0$. They satisfy
$$
u_{i0}'(y_{i0}(t))=\mu_i(t),\qquad
\beta^n u_{in}'(y_{in}(t))=\mu_i(t)\bar p_n\quad \text{for } n=1,\dots,N,
$$
and
$$
y_{i0}(t)+\sum_{n=1}^N \bar p_n y_{in}(t)=\bar w_i+t.
$$

Differentiating at $t=0$, we have
$$
\frac{d y_{i0}}{dt}(0)=\frac{1}{u_{i0}''(\bar x_{i0})}\frac{d\mu_i}{dt}(0)
=-\frac{\bar r_{i0}}{\bar\lambda_i}\frac{d\mu_i}{dt}(0),
$$
and, for each $n=1,\dots,N$,
$$
\frac{d y_{in}}{dt}(0)
=
\frac{\bar p_n}{\beta^n u_{in}''(\bar x_{in})}\frac{d\mu_i}{dt}(0)
=
-\frac{\bar r_{in}}{\bar\lambda_i}\frac{d\mu_i}{dt}(0).
$$
Differentiating the budget equation then yields
$$
-\frac{\bar r_i}{\bar\lambda_i}\frac{d\mu_i}{dt}(0)=1.
$$
Thus
$$
\frac{d\mu_i}{dt}(0)=-\frac{\bar\lambda_i}{\bar r_i},
$$
and hence
\begin{equation}\label{eq:proof-fixed-price}
	\frac{\partial \xi_{in}}{\partial w}(\bar p,\bar w_i)
	=
	\frac{d y_{in}}{dt}(0)
	=
	\frac{\bar r_{in}}{\bar r_i}.
\end{equation}

Finally, since $w_i(p)=\omega_{i0}+\sum_{k=1}^N p_k\omega_{ik}$, the chain rule gives
$$
\frac{d \xi_{in}}{d p_m}(\bar p)
=
\frac{\partial \xi_{in}}{\partial p_m}(\bar p,\bar w_i)
+
\frac{\partial \xi_{in}}{\partial w}(\bar p,\bar w_i)\omega_{im}.
$$
Using \eqref{eq:proof-fixed-wealth} and \eqref{eq:proof-fixed-price}, we obtain
$$
\frac{d \xi_{in}}{d p_m}(\bar p)
=
\left[
\frac{\bar r_{in}\bar r_{im}}{\bar r_i}
-\frac{\bar r_{in}}{\bar p_n}\mathbf 1_{\{n=m\}}
\right]
+
\frac{\bar r_{in}}{\bar r_i}(\omega_{im}-\bar x_{im}),
$$
which proves \eqref{eq:exchange-slutsky}.

\subsection{Proof of Lemma \ref{lem:substitutionandincomeeffect}}\label{proof:substitutionincomeeffect}

We separately study the quadratic form generated by the individual substitution matrix $S_i$ and the individual income-effect matrix $M_i$, defined in \eqref{eq:exchange-slutsky}. First, we have
$$
q^T S_i q
=
\frac{1}{\bar r_i}\left(\sum_{n=1}^N \bar r_{in}q_n\right)^2
-
\sum_{n=1}^N \frac{\bar r_{in}}{\bar p_n}q_n^2
=
\frac{(\bar r_i^0)^2}{\bar r_i}\Lambda_i(q)^2
-
\sum_{n=1}^N \frac{\bar r_{in}}{\bar p_n}q_n^2.
$$
The first equality is by definition, and the second uses $\sum_{n=1}^N \bar r_{in}q_n = \bar r_i^0\Lambda_i(q)$.
Next,
$$
\sum_{n=1}^N \frac{\bar r_{in}}{\bar p_n}\bigl(q_n-\bar p_n\Lambda_i(q)\bigr)^2
=
\sum_{n=1}^N \frac{\bar r_{in}}{\bar p_n}q_n^2
-
2\Lambda_i(q)\sum_{n=1}^N \bar r_{in}q_n
+
\Lambda_i(q)^2\sum_{n=1}^N \bar p_n\bar r_{in}.
$$
Using again $\sum_{n=1}^N \bar r_{in}q_n = \bar r_i^0\Lambda_i(q)$ and $\sum_{n=1}^N \bar p_n\bar r_{in}=\bar r_i^0$, we get
$$
\sum_{n=1}^N \frac{\bar r_{in}}{\bar p_n}\bigl(q_n-\bar p_n\Lambda_i(q)\bigr)^2
=
\sum_{n=1}^N \frac{\bar r_{in}}{\bar p_n}q_n^2
-
\bar r_i^0\Lambda_i(q)^2.
$$
Substituting this into the previous expression yields
$$
q^T S_i q
=
-\sum_{n=1}^N \frac{\bar r_{in}}{\bar p_n}\bigl(q_n-\bar p_n\Lambda_i(q)\bigr)^2
+
\left(\frac{(\bar r_i^0)^2}{\bar r_i}-\bar r_i^0\right)\Lambda_i(q)^2.
$$
Finally, since $\bar r_i=\bar r_{i0}+\bar r_i^0$,
$$
\frac{(\bar r_i^0)^2}{\bar r_i}-\bar r_i^0
=
\bar r_i^0\left(\frac{\bar r_i^0-\bar r_i}{\bar r_i}\right)
=
-\frac{\bar r_{i0}\bar r_i^0}{\bar r_i}.
$$
Therefore,
$$
q^T S_i q
=
-\sum_{n=1}^N \frac{\bar r_{in}}{\bar p_n}\bigl(q_n-\bar p_n\Lambda_i(q)\bigr)^2
-
\frac{\bar r_{i0}\bar r_i^0}{\bar r_i}\Lambda_i(q)^2.
$$

For the income-effect matrix, we have
$$
q^T M_i q
=
\sum_{n=1}^N \sum_{m=1}^N
\frac{\bar r_{in}q_n}{\bar r_i}(\omega_{im}-\bar x_{im})q_m
=
\frac{1}{\bar r_i}
\left(\sum_{n=1}^N \bar r_{in}q_n\right)
\left(\sum_{m=1}^N (\omega_{im}-\bar x_{im})q_m\right).
$$
Using $\sum_{n=1}^N \bar r_{in}q_n=\bar r_i^0\Lambda_i(q)$, we obtain
$$
q^T M_i q
=
\frac{\bar r_i^0}{\bar r_i}\Lambda_i(q)\sum_{m=1}^N (\omega_{im}-\bar x_{im})q_m.
$$
This proves the lemma.

\subsection{Proof of Proposition \ref{prop:quadratic-decomposition}} \label{proof:quadratic-decomposition}

For each agent $i\in I_N$, denote $J_i=S_i+M_i$. By Lemma \ref{lem:substitutionandincomeeffect}, 
$$
q^T J_i q
=
-\sum_{n=1}^N \frac{\bar r_{in}}{\bar p_n}\bigl(q_n-\bar p_n\Lambda_i(q)\bigr)^2
-
\frac{\bar r_{i0}\bar r_i^0}{\bar r_i}\Lambda_i(q)^2
+
\frac{\bar r_i^0}{\bar r_i}\Lambda_i(q)\sum_{m=1}^N (\omega_{im}-\bar x_{im})q_m.
$$
Write $q=\alpha \bar p+u$. Since $\Lambda_i$ is linear and $\Lambda_i(\bar p)=1$, we have
$\Lambda_i(q)=\Lambda_i(\alpha \bar p+u)=\alpha+\Lambda_i(u).$ Also,
$$
q_n-\bar p_n\Lambda_i(q)
=
\alpha \bar p_n+u_n-\bar p_n\bigl(\alpha+\Lambda_i(u)\bigr)
=
u_n-\bar p_n\Lambda_i(u).
$$
Hence, the first term becomes
$$
-\sum_{n=1}^N \frac{\bar r_{in}}{\bar p_n}\bigl(q_n-\bar p_n\Lambda_i(q)\bigr)^2
=
-\sum_{n=1}^N \frac{\bar r_{in}}{\bar p_n}\bigl(u_n-\bar p_n\Lambda_i(u)\bigr)^2.
$$
For the income-effect term,
$$
\sum_{m=1}^N (\omega_{im}-\bar x_{im})q_m
=
\sum_{m=1}^N (\omega_{im}-\bar x_{im})(\alpha \bar p_m+u_m)
=
\alpha \sum_{m=1}^N \bar p_m(\omega_{im}-\bar x_{im})
+
\sum_{m=1}^N (\omega_{im}-\bar x_{im})u_m.
$$
By agent $i$'s budget equality at equilibrium,
$\sum_{m=1}^N \bar p_m(\omega_{im}-\bar x_{im})
=
\bar x_{i0}-\omega_{i0}.$ Therefore,
$$
\sum_{m=1}^N (\omega_{im}-\bar x_{im})q_m
=
\alpha(\bar x_{i0}-\omega_{i0})
+
\sum_{m=1}^N (\omega_{im}-\bar x_{im})u_m.
$$
Substituting these identities into the expression for $q^T J_i q$, we obtain
\begin{align*}
	q^T J_i q
	={}&
	-\sum_{n=1}^N \frac{\bar r_{in}}{\bar p_n}\bigl(u_n-\bar p_n\Lambda_i(u)\bigr)^2
	-
	\frac{\bar r_{i0}\bar r_i^0}{\bar r_i}\bigl(\alpha+\Lambda_i(u)\bigr)^2 \\
	&\quad
	+
	\frac{\bar r_i^0}{\bar r_i}\bigl(\alpha+\Lambda_i(u)\bigr)
	\left[
	\alpha(\bar x_{i0}-\omega_{i0})
	+
	\sum_{m=1}^N (\omega_{im}-\bar x_{im})u_m
	\right].
\end{align*}
Expanding the last two terms gives
\begin{align*}
	q^T J_i q
	={}&
	-\sum_{n=1}^N \frac{\bar r_{in}}{\bar p_n}\bigl(u_n-\bar p_n\Lambda_i(u)\bigr)^2
	-
	\frac{\bar r_i^0}{\bar r_i}\bigl(\bar r_{i0}+\omega_{i0}-\bar x_{i0}\bigr)\alpha^2 \\
	&
	+
	\alpha\,\frac{\bar r_i^0}{\bar r_i}
	\left[
	\sum_{m=1}^N (\omega_{im}-\bar x_{im})u_m
	-
	\bigl(2\bar r_{i0}+\omega_{i0}-\bar x_{i0}\bigr)\Lambda_i(u)
	\right] \\
	&
	+
	\frac{\bar r_i^0}{\bar r_i}\Lambda_i(u)
	\left[
	\sum_{m=1}^N (\omega_{im}-\bar x_{im})u_m-\bar r_{i0}\Lambda_i(u)
	\right].
\end{align*}

Aggregating over $i\in I_N$, the first term becomes $-S(u)$, the second becomes $-A\alpha^2$, and the fourth becomes $M(u)$. The third term can be simplified: because
$u\in\ker \Psi$, 
$$\sum_{i\in I_N}\frac{\bar r_i^0}{\bar r_i}
\bigl(2\bar r_{i0}+\omega_{i0}-\bar x_{i0}\bigr)\Lambda_i(u)
=
\Psi(u)
=
0.$$
Therefore,  after aggregation over $i\in I_N$, the third term becomes,
$$
\alpha\sum_{i\in I_N}\frac{\bar r_i^0}{\bar r_i}\sum_{m=1}^N (\omega_{im}-\bar x_{im})u_m
=
\alpha\sum_{i\in I_N}\sum_{m=1}^N (\omega_{im}-\bar x_{im})u_m
-
\alpha\sum_{i\in I_N}\frac{\bar r_{i0}}{\bar r_i}\sum_{m=1}^N (\omega_{im}-\bar x_{im})u_m,
$$
where the equality here uses $\bar r_i^0=\bar r_i-\bar r_{i0}$.
By market clearing, $\sum_{i\in I_N}(\omega_{im}-\bar x_{im})=0$ for every $m$, so the first term vanishes. Therefore,
$$
\alpha\sum_{i\in I_N}\frac{\bar r_i^0}{\bar r_i}\sum_{m=1}^N (\omega_{im}-\bar x_{im})u_m
=
\alpha\sum_{i\in I_N}\frac{\bar r_{i0}}{\bar r_i}\sum_{m=1}^N (\bar x_{im}-\omega_{im})u_m
=
\alpha R(u).
$$
This proves
$$
q^T Dz(\bar{p}) q=-A\alpha^2+R(u)\alpha-S(u)+M(u).
$$

It remains to prove the last claim. Evaluating $\Psi$ at $\bar p$ and using $\Lambda_i(\bar p)=1$, we get
$$
\Psi(\bar p)
=
\sum_{i\in I_N}\frac{\bar r_i^0}{\bar r_i}\bigl(2\bar r_{i0}+\omega_{i0}-\bar x_{i0}\bigr)
=
A+\sum_{i\in I_N}\frac{\bar r_i^0}{\bar r_i}\bar r_{i0}.
$$
If $A>0$, then the right-hand side is strictly positive, because $\bar r_i^0>0$, $\bar r_i>0$, and $\bar r_{i0}>0$ for every $i\in I_N$. Thus $\Psi(\bar p)>0$, so $\bar p\notin\ker\Psi$. Since $\Psi$ is a linear functional, $\ker\Psi$ is a hyperplane, and it follows that
$\mathbb R^N=\operatorname{span}\{\bar p\}\oplus \ker\Psi.$ Hence, every $q\in\mathbb R^N$ admits a unique decomposition $q=\alpha\bar p+u$ with $\Psi(u)=0$.

\subsection{Proof of Proposition \ref{prop:equilibrium-bounds}}\label{proof:equilibrium-bounds}

Throughout this proof, we fix $\beta\in(0,1)$. Set $q_n=\bar{p}_n/\beta^n$ for each $n=1,\dots,N$ and $q_0=1$. Denote the equilibrium wealth of agent $i$ by $W_i=\omega_{i0}+\sum_{n=1}^N \bar{p}_n\omega_{in}.$ Moreover, define the shadow value of money of agent $i$ by $\lambda_i=u_{i0}'(\bar{x}_{i0}).$ The first-order conditions imply $u_{in}'(\bar{x}_{in})=\lambda_iq_n$ for each $i\in I_N$ and each $n=0,\dots,N$.

First, we argue that $u_{in}'(x)$ is uniformly bounded by two functions as a consequence of Assumption \textnormal{(A2)}. Since multiplying agent $i$'s utility by a positive constant does not change her demand, we normalize $u_{i0}'(1)=1$ for every $i$. Then Assumption \textnormal{(A2)} implies $u_{in}'(1)\in[c_u,C_u]$ for every $i$ and $n$. Note that $-r_{in}(x)^{-1}=\frac{u_{in}''(x)}{u_{in}'(x)}=\frac{d\log u_{in}'(x)}{dx}$, by Assumption \textnormal{(A2)}, we have
$$-\frac{1}{c_ux}\le \frac{d\log u_{in}'(x)}{dx}\le -\frac{1}{C_ux}.$$
For $x>1$, integrating from $1$ to $x$, we have $\frac{-\log x}{c_u}\le \log u_{in}'(x)- \log u_{in}'(1)\le \frac{-\log x}{C_u}$; similarly, for $x<1$, integrating from $1$ to $x$, we have $\frac{-\log x}{c_u}\ge \log u_{in}'(x)- \log u_{in}'(1)\ge  \frac{-\log x}{C_u}$. Therefore, using $u_{in}'(1)\in [c_u, C_u]$, we have
$$L(x)\le u_{in}'(x)\le U(x)$$
for all $x>0$, where $L(x)=c_u\min\{x^{-1/c_u},x^{-1/C_u}\}$ and $U(x)=C_u\max\{x^{-1/c_u},x^{-1/C_u}\}$. One observation we will repeatedly use is that $L$ and $U$ are decreasing, so $u_{in}'(x)\ge U(a)$ implies $x\le a$ and $u_{in}'(x)\le L(a)$ implies $x\ge a$.

Now, we show that $q_n$ is uniformly bounded by positive numbers. Fix $\gamma\in(0,1)$ such that $\gamma C_\Omega<c_\Omega$. By the curvature bound in Assumption \textnormal{(A2)},
$$
\log \frac{u_{i0}'(\gamma x)}{u_{i0}'(x)}
=
\int_{\gamma x}^{x} -\frac{u_{i0}''(t)}{u_{i0}'(t)}dt
\le
\int_{\gamma x}^{x}\frac{1}{c_ut}dt
=
\frac{1}{c_u}\log\frac{1}{\gamma}.
$$
Therefore, $\frac{u_{i0}'(\gamma x)}{u_{i0}'(x)}\le \gamma^{-1/c_u}$ for every $i$ and every $x>0$. If $\bar x_{in}\ge \gamma \bar x_{i0}$, then, by the first-order conditions and the marginal-utility comparability in Assumption \textnormal{(A2)},
$$
q_n
=
\frac{u_{in}'(\bar x_{in})}{u_{i0}'(\bar x_{i0})}
\le
C_u
\frac{u_{i0}'(\bar x_{in})}{u_{i0}'(\bar x_{i0})}
\le
C_u
\frac{u_{i0}'(\gamma \bar x_{i0})}{u_{i0}'(\bar x_{i0})}
\le
C_u\gamma^{-1/c_u}.
$$
Therefore, if $q_n>C_u\gamma^{-1/c_u}$, then $\bar x_{in}<\gamma\bar x_{i0}$ for every $i\in I_N$. Summing over agents and using market clearing,
$$
\Omega_n^N
=
\sum_{i\in I_N}\bar x_{in}
<
\gamma\sum_{i\in I_N}\bar x_{i0}
=
\gamma\Omega_0^N
\le
\gamma C_\Omega I_N
<
c_\Omega I_N,
$$
which contradicts Assumption \textnormal{(A3)}. Hence $q_n\le C_p$ for some $C_p<\infty$.

For the lower bound, choose $\Gamma>1$ such that $\Gamma c_\Omega>C_\Omega$. Again by the curvature bound in Assumption \textnormal{(A2)}, we have $\frac{u_{i0}'(\Gamma x)}{u_{i0}'(x)}\ge \Gamma^{-1/c_u}$ for every $i$ and every $x>0$. If $\bar x_{in}\le \Gamma\bar x_{i0}$, then, by the first-order conditions and the marginal-utility comparability in Assumption \textnormal{(A2)},
$$
q_n
=
\frac{u_{in}'(\bar x_{in})}{u_{i0}'(\bar x_{i0})}
\ge
c_u
\frac{u_{i0}'(\bar x_{in})}{u_{i0}'(\bar x_{i0})}
\ge
c_u
\frac{u_{i0}'(\Gamma\bar x_{i0})}{u_{i0}'(\bar x_{i0})}
\ge
c_u\Gamma^{-1/c_u}.
$$
Therefore, if $q_n<c_u\Gamma^{-1/c_u}$, then $\bar x_{in}>\Gamma\bar x_{i0}$ for every $i\in I_N$. Summing over agents and using market clearing,
$$
\Omega_n^N
=
\sum_{i\in I_N}\bar x_{in}
>
\Gamma\sum_{i\in I_N}\bar x_{i0}
=
\Gamma\Omega_0^N
\ge
\Gamma c_\Omega I_N
>
C_\Omega I_N,
$$
which again contradicts Assumption \textnormal{(A3)}. Thus, for some constants $0<c_p<C_p<\infty$,
$c_p\le q_n=\bar p_n/\beta^n\le C_p$ holds for every $n=1,\dots,N$.

Next, we show that the shadow values of money of all agents are uniformly bounded whenever $N_\beta$ is sufficiently large. By Assumption \textnormal{(A4)} and the price bounds just established, individual wealth is of order $N_\beta$:
$$
c_pc_WN_\beta \le W_i\le C_W+C_pC_WN_\beta.
$$
Indeed, the lower bound follows from $W_i\ge \sum_{n=1}^N \bar p_n\omega_{in}\ge c_p\sum_{n=1}^N\beta^n\omega_{in}\ge c_pc_WN_\beta$, and the upper bound follows from $W_i\le C_W+C_p\sum_{n=1}^N\beta^n C_W=C_W+C_pC_WN_\beta$.

Choose $\delta>0$ and a cutoff $N_0\in\mathbb N$ so that $\delta(1+C_pN_\beta)<c_pc_WN_\beta$ whenever $N_\beta\ge N_0$. Moreover, take $\bar\lambda>0$ such that $\bar\lambda\min\{1,c_p\}\ge U(\delta).$ If $\lambda_i>\bar\lambda$, then, for every $n=0,1,\dots,N$,
$$u_{in}'(\bar x_{in})=\lambda_iq_n\ge \bar\lambda\min\{1,c_p\}\ge U(\delta).$$
Hence, $\bar x_{in}\le\delta$ for every $n=0,1,\dots,N$. In this case, by the budget equality,
$$
W_i
=
\bar x_{i0}+\sum_{n=1}^N \bar p_n\bar x_{in}
\le
\delta\left(1+\sum_{n=1}^N \bar p_n\right)
\le
\delta(1+C_pN_\beta)
<
c_pc_WN_\beta,
$$
contradicting the lower bound on $W_i$. Therefore, $\lambda_i\le\bar\lambda$.

Conversely, choose $B>0$ and increase $N_0$ if necessary, so that $B(1+c_pN_\beta)>C_W+C_pC_WN_\beta$ whenever $N_\beta\ge N_0$. Take $\underline\lambda>0$ such that $\underline\lambda\max\{1,C_p\}\le L(B).$ If $\lambda_i<\underline\lambda$, then, for every $n=0,1,\dots,N$,
$$
u_{in}'(\bar x_{in})=\lambda_iq_n\le \underline\lambda\max\{1,C_p\}\le L(B).
$$ 
Hence, $\bar x_{in}\ge B$ for every $n=0,1,\dots,N$. In this case,
$$
W_i
=
\bar x_{i0}+\sum_{n=1}^N \bar p_n\bar x_{in}
\ge
B\left(1+\sum_{n=1}^N \bar p_n\right)
\ge
B(1+c_pN_\beta)
>
C_W+C_pC_WN_\beta,
$$
contradicting the upper bound on $W_i$. Therefore, after possibly increasing $N_0$, we have 
$\underline\lambda\le \lambda_i\le \bar\lambda$ for every $i\in I_N$ whenever $N_\beta\ge N_0$.

Now, we complete the proof. Choose $c_x,C_x>0$ such that
$$
L(c_x)\ge \bar\lambda\max\{1,C_p\},
\qquad
U(C_x)\le \underline\lambda\min\{1,c_p\}.
$$
Using again
$u_{in}'(\bar x_{in})=\lambda_i q_n
\in
[\underline\lambda\min\{1,c_p\},\bar\lambda\max\{1,C_p\}],$
we know $\bar x_{in}\in[c_x,C_x]$ for all $i\in I_N$ and $n=0,\dots,N$. Finally, Assumption \textnormal{(A2)} gives
$c_u\bar x_{in}\le \bar r_{in}\le C_u\bar x_{in},$ which yields
$c_r\le \bar r_{in}\le C_r$
for all $i\in I_N$ and $n=0,\dots,N$. This completes the proof.

\subsection{Proof of Lemma \ref{lem:A-positive}}\label{proof:A-positive}
We take $N_\beta\ge N_0$. For each agent $i\in I_N$, since $\bar r_i^0=\bar r_i-\bar r_{i0}$ by definition,
$$
A
=
\sum_{i\in I_N}\frac{\bar r_i^0}{\bar r_i}\bigl(\bar r_{i0}+\omega_{i0}-\bar x_{i0}\bigr)
=
\sum_{i\in I_N}\bigl(\bar r_{i0}+\omega_{i0}-\bar x_{i0}\bigr)
-
\sum_{i\in I_N}\frac{\bar r_{i0}}{\bar r_i}\bigl(\bar r_{i0}+\omega_{i0}-\bar x_{i0}\bigr).
$$
By market clearing at date $0$, $\sum_{i\in I_N}(\omega_{i0}-\bar x_{i0})=0$. Hence
$A
=
\sum_{i\in I_N}\bar r_{i0}
-
\sum_{i\in I_N}\frac{\bar r_{i0}}{\bar r_i}\bigl(\bar r_{i0}+\omega_{i0}-\bar x_{i0}\bigr).$ By Proposition \ref{prop:equilibrium-bounds}, we have $\bar r_{i0}\ge c_r$ for every $i\in I_N$. Therefore, $\sum_{i\in I_N}\bar r_{i0}\ge c_r I_N.$

On the other hand, Proposition \ref{prop:equilibrium-bounds} and Assumption \textnormal{(A4)} imply that $\bar r_{i0}$, $\bar x_{i0}$, and $\omega_{i0}$ are uniformly bounded. Hence there exists a constant $C>0$, independent of $i$ and $N$, such that
$\bigl|\bar r_{i0}+\omega_{i0}-\bar x_{i0}\bigr|\le C$ for all $i\in I_N$. Moreover, by Proposition \ref{prop:equilibrium-bounds}, there exists a constant $C'>0$ such that
$\frac{\bar r_{i0}}{\bar r_i}\le \frac{C'}{N_\beta}$ for all $i\in I_N$. It follows that
$\left|
\sum_{i\in I_N}\frac{\bar r_{i0}}{\bar r_i}\bigl(\bar r_{i0}+\omega_{i0}-\bar x_{i0}\bigr)
\right|
\le
\frac{CC'}{N_\beta} I_N.$

Choose $N_1\ge N_0$ so that $CC'/N_\beta\le c_r/2$ whenever $N_\beta\ge N_1$. In this case, $A\ge c_r I_N-\frac{CC'}{N_\beta}I_N\ge \frac{c_r}{2}I_N.$
Thus \eqref{eq:A-positive} holds with $c_A:=c_r/2$.

\subsection{Proof of Lemma \ref{lem:R-order-one}} \label{proof:R-order-one}
We take $N_\beta\ge N_0$. Write $R(u)=\sum_{n=1}^N \alpha_n v_n$ for 
$\alpha_n=\bar p_n\sum_{i\in I_N}\frac{\bar r_{i0}}{\bar r_i}(\bar x_{in}-\omega_{in})$ for each  $n=1,\dots,N.$ By Proposition \ref{prop:equilibrium-bounds}, there exists a constant $C_\eta>0$ such that $\frac{\bar r_{i0}}{\bar r_i}\le \frac{C_\eta}{N_\beta}$ holds for all $i\in I_N$. Hence, using Proposition \ref{prop:equilibrium-bounds},
$$
|\alpha_n|
\le
\frac{C_\eta C_p\beta^n}{N_\beta}\sum_{i\in I_N}|\bar x_{in}-\omega_{in}|\le C_\tau I_N \pi_n^{N,\beta},
$$
for some $C_\tau>0$, where the second inequality follows from Proposition \ref{prop:equilibrium-bounds} and Assumption \textnormal{(A3)}, since
$$
\sum_{i\in I_N}|\bar x_{in}-\omega_{in}|
\le
\sum_{i\in I_N}\bar x_{in}+\sum_{i\in I_N}\omega_{in}
\le
C_x I_N+\Omega_n^N
\le
(C_x+C_\Omega)I_N.
$$ 
Finally, by the Cauchy--Schwarz inequality,
$$
|R(u)|
=
\left|\sum_{n=1}^N \alpha_n v_n\right|
\le
\left(
\sum_{n=1}^N \frac{\alpha_n^2}{\pi_n^{N,\beta}}
\right)^{1/2}
\left(
\sum_{n=1}^N \pi_n^{N,\beta} v_n^2
\right)^{1/2}.
$$
Since $\alpha_n^2/\pi_n^{N,\beta}\le C_\tau^2 I_N^2 \pi_n^{N,\beta}$ and $\sum_{n=1}^N \pi_n^{N,\beta}=1$, we obtain
$|R(u)|
\le
C_\tau I_N
\left(
\sum_{n=1}^N \pi_n^{N,\beta} v_n^2
\right)^{1/2}.$ This proves the lemma with $C_R:=C_\tau$.

\subsection{The Isoelastic Example}
\label{appendix:isoelasticexample}

In this subsection, we compute the equilibrium objects of the isoelastic example in Section \ref{subsec:comparisons}. We write $p_0=1$, $q_n=p_n/\beta^n$, and $\mu_i=\lambda_i^{-\sigma}$. The first-order conditions give
$$
x_{i0}=\mu_i,\qquad
x_{in}=\mu_i(1+\delta\varepsilon_{in})q_n^{-\sigma},
\quad n=1,\dots,N.
$$
Market clearing implies
$$
\frac{1}{\bar\omega I_N}\sum_{i\in I_N} \mu_i=1,\qquad
q_n^\sigma
=
\frac{1}{\bar\omega I_N}\sum_{i\in I_N} \mu_i(1+\delta\varepsilon_{in}),
\quad n=1,\dots,N.
$$
By Proposition \ref{prop:equilibrium-bounds}, the $q_n$'s are uniformly bounded above and away from zero whenever $N_\beta$ is sufficiently large.

Substituting the first-order conditions into the budget equation, we have
$$
\mu_i\left[
1+\sum_{n=1}^N\beta^n(1+\delta\varepsilon_{in})q_n^{1-\sigma}
\right]
=
\bar\omega\left(1+\sum_{n=1}^N\beta^nq_n\right)
+\eta_i\left(s_0+\sum_{n=1}^N\beta^nq_ns_n\right).
$$
Define the remainder $err_i$ by
$$
\mu_i=
\frac{
	\bar\omega\left(1+\sum_{n=1}^N\beta^nq_n\right)
	+\eta_i\left(s_0+\sum_{n=1}^N\beta^nq_ns_n\right)}
{1+\sum_{n=1}^N\beta^nq_n^{1-\sigma}}
+err_i.
$$
Since the numerator and denominator are uniformly comparable to $1+N_\beta$,
one can show that the budget equation implies
\begin{equation}
	|err_i|
	\le
	\frac{C_1}{1+N_\beta}
	\left|
	\sum_{n=1}^N\beta^n\varepsilon_{in}q_n^{1-\sigma}
	\right|
	=
	\frac{C_1}{1+N_\beta}
	\left|
	\sum_{n=1}^N\beta^n\varepsilon_{in}
	(q_n^{1-\sigma}-1)
	\right|
	\le
	\frac{C_2}{1+N_\beta}
	\sum_{n=1}^N\beta^n |q_n^{1-\sigma}-1|,
	\label{eqn:err}
\end{equation}
for some constants $C_1,C_2>0$, where the equality uses
$\sum_{n=1}^N\beta^n\varepsilon_{in}=0$ and the last inequality uses the uniform boundedness of $\varepsilon_{in}$.

The fraction defining $\mu_i-err_i$ is affine in $\eta_i$. Hence, the assumptions on $\varepsilon_{in}$ and orthogonality between $\varepsilon_{in}$ and $\eta_i$ imply 
$$
\frac{1}{I_N}\sum_{i\in I_N} \mu_i\varepsilon_{in}
=
\frac{1}{I_N}\sum_{i\in I_N} err_i\varepsilon_{in}.
$$
Combining this with market clearing gives
$$
q_n^\sigma-1
=
\frac{\delta}{\bar\omega I_N}\sum_{i\in I_N} err_i\varepsilon_{in}.
$$
Since the $q_n$'s are uniformly bounded above and away from zero, and the map $y\mapsto y^{(1-\sigma)/\sigma}$ is Lipschitz on the relevant compact interval bounding $q_n$,
$$
|q_n^{1-\sigma}-1|
=
\left|(q_n^\sigma)^{(1-\sigma)/\sigma}-1\right|
\le C_3|q_n^\sigma-1|
\le
C_4\left|
\frac{1}{I_N}\sum_{i\in I_N} err_i\varepsilon_{in}
\right|, 
$$
for $C_3,C_4>0$. Thus,
$$
\frac{1}{N_\beta}\sum_{n=1}^N
\beta^n(q_n^{1-\sigma}-1)^2
\le
\frac{C_4^2}{I_N^2}\sum_{i,j\in I_N} err_i err_j
\left(
\frac{1}{N_\beta}\sum_{n=1}^N \beta^n \varepsilon_{in}\varepsilon_{jn}
\right)\le
C_4^2 U_{N,\beta}^2 \kappa_{N,\beta},
$$
where
$U_{N,\beta}=\max_{i\in I_N}|err_i|,$ and 
$\kappa_{N,\beta}
=
\frac{1}{I_N^2}\sum_{i,j\in I_N}
\left|
\frac{1}{N_\beta}\sum_{n=1}^N \beta^n \varepsilon_{in}\varepsilon_{jn}
\right|.$
By assumption, $\kappa_{N,\beta}\to0$ as $N_\beta\to\infty$.

Now, by \eqref{eqn:err} and Cauchy--Schwarz, we have
$$
|err_i|
\le
\frac{C_2N_\beta}{1+N_\beta}
\left(
\frac{1}{N_\beta}\sum_{n=1}^N \beta^n (q_n^{1-\sigma}-1)^2
\right)^{1/2}
\le
C_5 U_{N,\beta} \kappa_{N,\beta}^{1/2},
$$
for some $C_5>0$. Taking the maximum over $i\in I_N$, we have
$U_{N,\beta}\le C_5 U_{N,\beta}\kappa_{N,\beta}^{1/2}.$ Since $\kappa_{N,\beta}\to0$, for all sufficiently large $N_\beta$ we have $C_5\kappa_{N,\beta}^{1/2}<1$. Therefore $U_{N,\beta}=0$, which implies $err_i=0$ for every $i\in I_N$. Hence,
$q_n^\sigma-1
=
\frac{\delta}{\bar\omega I_N}\sum_{i\in I_N}err_i\varepsilon_{in}
=
0,$
which implies $q_n=1$, and hence $p_n=\beta^n$, for every $n=1,\dots,N$.

With $q_n=1$, the budget equation and the discounted balance condition yield
$\mu_i=\bar\omega+\eta_i s^{N,\beta},$
where $s^{N,\beta}:=\frac{s_0+\sum_{n=1}^N\beta^n s_n}{1+N_\beta}.$
Thus, the allocation is given by
$$
x_{i0}=\bar\omega+\eta_i s^{N,\beta},\qquad
x_{in}=(\bar\omega+\eta_i s^{N,\beta})(1+\delta\varepsilon_{in}),
\quad n=1,\dots,N,
$$
and the shadow value of money is given by
$$
\lambda_i=(\mu_i)^{-1/\sigma}
=
(\bar\omega+\eta_i s^{N,\beta})^{-1/\sigma}.
$$

\section*{Acknowledgments}

This paper grew out of a chapter of my Yale PhD thesis. I am especially grateful to Truman Bewley and John Geanakoplos for their encouragement and advice. I have also benefited from helpful comments by Alberto Bisin, Andrei Gomberg, Michael Greinecker, Ali Khan, Xiangliang Li, Anna Rubinchik, Larry Samuelson, Eddie Schlee, Keisuke Teeple, Yakun Xi, Michael Zierhut, and workshop participants at the Central European Program in Economic Theory.

\addcontentsline{toc}{section}{References}

\bibliographystyle{abbrvnat}

\bibliography{RS}

@article{bewley1972existence,
 author  = {Bewley, Truman F.},
 title   = {Existence of equilibria in economies with infinitely many commodities},
 journal = {J. Econ. Theory},
 volume  = {4},
 number  = {3},
 pages   = {514--540},
 year    = {1972}
}

@article{stroyan1983myopic,
 author  = {Stroyan, K. D.},
 title   = {Myopic utility functions on sequential economies},
 journal = {J. Math. Econ.},
 volume  = {11},
 number  = {3},
 pages   = {267--276},
 year    = {1983}
}

@article{debreu1974excess,
 author  = {Debreu, Gerard},
 title   = {Excess demand functions},
 journal = {J. Math. Econ.},
 volume  = {1},
 number  = {1},
 pages   = {15--21},
 year    = {1974}
}

@article{vives1987small,
 author  = {Vives, Xavier},
 title   = {Small income effects: A {Marshallian} theory of consumer surplus and downward sloping demand},
 journal = {Rev. Econ. Stud.},
 volume  = {54},
 number  = {1},
 pages   = {87--103},
 year    = {1987}
}

@article{grandmont1992transformations,
 author  = {Grandmont, Jean-Michel},
 title   = {Transformations of the commodity space, behavioral heterogeneity, and the aggregation problem},
 journal = {J. Econ. Theory},
 volume  = {57},
 number  = {1},
 pages   = {1--35},
 year    = {1992}
}

@article{dierker1972number,
 author  = {Dierker, Egbert},
 title   = {Two remarks on the number of equilibria of an economy},
 journal = {Econometrica},
 volume  = {40},
 number  = {5},
 pages   = {951--953},
 year    = {1972}
}

@article{mantel1974characterization,
 author  = {Mantel, Rolf R.},
 title   = {On the characterization of aggregate excess demand},
 journal = {J. Econ. Theory},
 volume  = {7},
 number  = {3},
 pages   = {348--353},
 year    = {1974}
}

@article{quah1997law,
 author  = {Quah, John K.-H.},
 title   = {The law of demand when income is price dependent},
 journal = {Econometrica},
 volume  = {65},
 number  = {6},
 pages   = {1421--1442},
 year    = {1997}
}

@article{GeanakoplosWalsh2018,
 author  = {Geanakoplos, John and Walsh, Kieran James},
 title   = {Uniqueness and stability of equilibrium in economies with two goods},
 journal = {J. Econ. Theory},
 volume  = {174},
 pages   = {261--272},
 year    = {2018}
}

@article{kehoe1991gross,
 author  = {Kehoe, Timothy J. and Levine, David K. and Mas-Colell, Andreu and Woodford, Michael},
 title   = {Gross substitutability in large-square economies},
 journal = {J. Econ. Theory},
 volume  = {54},
 number  = {1},
 pages   = {1--25},
 year    = {1991}
}

@article{hayashi2013smallness,
 author  = {Hayashi, Takashi},
 title   = {Smallness of a commodity and partial equilibrium analysis},
 journal = {J. Econ. Theory},
 volume  = {148},
 number  = {1},
 pages   = {279--305},
 year    = {2013}
}

@article{hildenbrand1983law,
 author  = {Hildenbrand, Werner},
 title   = {On the law of demand},
 journal = {Econometrica},
 volume  = {51},
 number  = {4},
 pages   = {997--1019},
 year    = {1983}
}

@techreport{quah2004aggregate,
 author      = {Quah, John K.-H.},
 title       = {The aggregate weak axiom in a financial economy through dominant substitution effects},
 institution = {Economics Group, Nuffield College, University of Oxford},
 type        = {Economics Papers},
 number      = {2004-W18},
 year        = {2004}
}

@article{scarf1960some,
 author  = {Scarf, Herbert},
 title   = {Some examples of global instability of the competitive equilibrium},
 journal = {Int. Econ. Rev.},
 volume  = {1},
 number  = {3},
 pages   = {157--172},
 year    = {1960}
}

@article{bewley1980permanent1,
 author  = {Bewley, Truman},
 title   = {The permanent income hypothesis and short-run price stability},
 journal = {J. Econ. Theory},
 volume  = {23},
 number  = {3},
 pages   = {323--333},
 year    = {1980}
}

@article{jerison1999dispersed,
 author  = {Jerison, Michael},
 title   = {Dispersed excess demands, the weak axiom and uniqueness of equilibrium},
 journal = {J. Math. Econ.},
 volume  = {31},
 number  = {1},
 pages   = {15--48},
 year    = {1999}
}

@article{brown1981myopic,
 author  = {Brown, Donald J. and Lewis, Lucinda M.},
 title   = {Myopic economic agents},
 journal = {Econometrica},
 volume  = {49},
 number  = {2},
 pages   = {359--368},
 year    = {1981}
}

@article{sonnenschein1972market,
 author  = {Sonnenschein, Hugo},
 title   = {Market excess demand functions},
 journal = {Econometrica},
 volume  = {40},
 number  = {3},
 pages   = {549--563},
 year    = {1972}
}

@article{polterovich1978criteria,
 author  = {Mitjushin, Leonid G. and Polterovich, Victor M.},
 title   = {Criteria for monotonicity of demand functions},
 journal = {Ekonomika i Matematicheskie Metody},
 volume  = {14},
 number  = {1},
 pages   = {122--128},
 year    = {1978},
 note    = {In Russian}
}

@article{arrow1959stability,
 author  = {Arrow, Kenneth J. and Block, H. D. and Hurwicz, Leonid},
 title   = {On the stability of the competitive equilibrium, {II}},
 journal = {Econometrica},
 volume  = {27},
 number  = {1},
 pages   = {82--109},
 year    = {1959}
}

@article{toda2024recent,
 author  = {Toda, Alexis Akira and Walsh, Kieran James},
 title   = {Recent advances on uniqueness of competitive equilibrium},
 journal = {J. Math. Econ.},
 volume  = {113},
 pages   = {103008},
 year    = {2024}
}

@misc{lorenzoniwerning2025tatonnement,
 author       = {Lorenzoni, Guido and Werning, Iv{\'a}n},
 title        = {T{\^a}tonnement and price setting in general equilibrium},
 howpublished = {Slides presented at the {NBER} Summer Institute, Impulse and Propagation Mechanisms},
 note         = {July 9, 2025},
 url          = {https://conference.nber.org/conf_papers/f227101/f227101.slides.pdf},
 year         = {2025}
}

@article{bastianello2017topological,
 author  = {Bastianello, Lorenzo},
 title   = {A topological approach to delay aversion},
 journal = {J. Math. Econ.},
 volume  = {73},
 pages   = {1--12},
 year    = {2017}
}

@article{quah2000monotonicity,
 author  = {Quah, John K.-H.},
 title   = {The monotonicity of individual and market demand},
 journal = {Econometrica},
 volume  = {68},
 number  = {4},
 pages   = {911--930},
 year    = {2000}
}

@article{grandmont1987distributions,
 author  = {Grandmont, Jean-Michel},
 title   = {Distributions of preferences and the law of demand},
 journal = {Econometrica},
 volume  = {55},
 number  = {1},
 pages   = {155--161},
 year    = {1987}
}

@article{hayashi2008note,
 author  = {Hayashi, Takashi},
 title   = {A note on small income effects},
 journal = {J. Econ. Theory},
 volume  = {139},
 number  = {1},
 pages   = {360--379},
 year    = {2008}
}

@techreport{weretka2018quasilinear,
 author = {Weretka, Marek},
 title  = {Quasilinear approximations},
 type   = {Working paper},
 year   = {2018}
}

@article{keenan2025global,
 author  = {Keenan, Donald C. and Kim, Taewon},
 title   = {Global stability in large economies},
 journal = {Econ. Theory Bull.},
 volume  = {13},
 number  = {2},
 pages   = {189--199},
 year    = {2025},
 doi     = {10.1007/s40505-025-00288-y}
}

@article{freixas1987engel,
 author  = {Freixas, Xavier and Mas-Colell, Andreu},
 title   = {Engel curves leading to the weak axiom in the aggregate},
 journal = {Econometrica},
 volume  = {55},
 number  = {3},
 pages   = {515--531},
 year    = {1987}
}

@misc{milleron1974unicite,
 author       = {Milleron, Jean-Claude},
 title        = {Unicit{\'e} et stabilit{\'e} de l'{\'e}quilibre en {\'e}conomie de distribution},
 howpublished = {S{\'e}minaire d'{\'E}conom{\'e}trie Roy-Malinvaud},
 year         = {1974}
}

@article{weretka2026welfare,
 author  = {Weretka, Marek and Dec, Marcin},
 title   = {Welfare measurements with heterogeneous agents},
 journal = {J. Econ. Dyn. Control},
 volume  = {184},
 pages   = {105252},
 year    = {2026},
 doi     = {10.1016/j.jedc.2025.105252}
}

@article{araujo1985lack,
 author  = {Araujo, Aloisio},
 title   = {Lack of {Pareto} optimal allocations in economies with infinitely many commodities: The need for impatience},
 journal = {Econometrica},
 volume  = {53},
 number  = {2},
 pages   = {455--461},
 year    = {1985}
}

@article{dana1995extension,
 author  = {Dana, Rose-Anne},
 title   = {An extension of {Milleron}, {Mitjushin} and {Polterovich}'s result},
 journal = {J. Math. Econ.},
 volume  = {24},
 number  = {3},
 pages   = {259--269},
 year    = {1995}
}

@incollection{mas1991uniqueness,
 author    = {Mas-Colell, Andreu},
 title     = {On the uniqueness of equilibrium once again},
 booktitle = {Equilibrium Theory and Applications},
 editor    = {Barnett, William A. and Cornet, Bernard and D'Aspremont, Claude and Gabszewicz, Jean and Mas-Colell, Andreu},
 publisher = {Cambridge University Press},
 pages     = {275--296},
 year      = {1991}
}

@incollection{arrow1959toward,
 author    = {Arrow, Kenneth J.},
 title     = {Toward a theory of price adjustment},
 booktitle = {The Allocation of Economic Resources: Essays in Honor of Bernard Francis Haley},
 editor    = {Abramovitz, Moses},
 publisher = {Stanford University Press},
 address   = {Stanford, CA},
 pages     = {41--51},
 year      = {1959}
}

\end{document}